\documentclass[11pt]{article}
\usepackage[left=1in, top=1in, right=1in, bottom=1in]{geometry}

\usepackage{amsmath, amssymb, amsthm, amsfonts}
\usepackage{mathabx}
\usepackage[numbers]{natbib}
\usepackage{float}
\usepackage{caption}
\usepackage{xfrac}
\usepackage[T1]{fontenc}
\usepackage{diagbox}
\usepackage{bbm}
\usepackage{thm-restate}
\usepackage{hyperref}

\usepackage{titlesec}
\titleformat{\paragraph}[runin]{\normalfont\normalsize\bfseries}{}{0em}{}[.]

\usepackage[ruled]{algorithm2e} 
\usepackage{algorithmic}

\SetAlFnt{\small}
\SetAlCapFnt{\small}
\SetAlCapNameFnt{\small}
\SetAlCapHSkip{0pt}
\IncMargin{-\parindent}

\newtheorem{theorem}{Theorem}[section]

\newtheorem{example}{Example}[section]
\newtheorem{corollary}[theorem]{Corollary}
\newtheorem{lemma}[theorem]{Lemma}
\newtheorem{proposition}[theorem]{Proposition}

\newtheorem{claim}[theorem]{Claim}
\newtheorem{definition}[theorem]{Definition}
\newtheorem{remark}[theorem]{Remark}
\newtheorem{observation}[theorem]{Observation}

\newcommand{\actions}{A}
\newcommand{\actionset}{2^\actions}
\newcommand{\set}{S}
\newcommand{\setal}{\set_\alpha}
\newcommand{\setalp}{\set_{\alpha'}}
\newcommand{\altset}{T}
\newcommand{\reals}{\mathbb{R}}
\newcommand{\demandset}{\mathcal{D}}
\newcommand{\orbit}[1]{\mathsf{C}_{#1}}
\newcommand{\utilmin}[1]{\underline{u}_p(#1)}
\newcommand{\successor}{\mathsf{succ}}
\newcommand{\g}{V}
\newcommand{\nextf}{\successor_{f,c}}
\newcommand{\cand}[2]{\Gamma_{#1}(#2)}
\newcommand{\poten}{\Phi}
\newcommand{\E}{\mathbb{E}}
\newcommand{\greedy}{\textsc{greedy}}
\newcommand{\nill}{\mathsf{null}}
\newcommand{\epsper}{$\epsilon$-perturbation}
\newcommand{\epsiper}{$\epsilon_i$-perturbation}

\DeclareMathOperator*{\argmax}{arg\,max}
\DeclareMathOperator*{\argmin}{arg\,min}

\newcommand{\orbitfc}{\orbit{f,c}}
\newcommand{\ratios}{Q}

\title{Combinatorial Contracts\thanks{An extended abstract appeared in the Proceedings of the 62nd IEEE Symposium on Foundations of Computer Science (FOCS 2021), Boulder, CO, USA. The FPTAS in Section~\ref{sec:ptas} first appeared in \cite{DuttingEFK25}, and strengthens a result that appeared in the extended abstract. This project has received funding from the European Research Council (ERC) under the European Union's Horizon 2020 research and innovation program (grant agreement No. 866132), by the NSF-BSF (grant number 2020788), by the Israel Science Foundation Breakthrough Program (grant No.~2600/24), and by a grant from TAU Center for AI and Data Science (TAD). T. Ezra is supported by the Harvard University Center of Mathematical Sciences and Applications.}}

\author{Paul D\"utting\thanks{Google Research, Switzerland. Email: \texttt{duetting@google.com}} \and Tomer Ezra\thanks{Harvard University, USA. Email: \texttt{tomer@cmsa.fas.harvard.edu}} \and Michal Feldman\thanks{Tel Aviv University and Microsoft ILDC, Israel. Email: \texttt{mfeldman@tau.ac.il}} \and Thomas Kesselheim\thanks{University of Bonn, Germany. Email: \texttt{thomas.kesselheim@uni-bonn.de}}}

\date{}

\begin{document}

\maketitle

\begin{abstract}
We introduce a new model of combinatorial contracts in which a principal delegates the execution of a costly task to an agent. To complete the task, the agent can take any subset of a given set of unobservable actions, each of which has an associated cost.
The cost of a set of actions is the sum of the costs of the individual actions, and
the principal's reward as a function of the chosen actions satisfies some form of diminishing returns.
The principal incentivizes the agents through a contract, based on the observed outcome.

Our main results are for the case where the task delegated to the agent is a project, which can be successful or not.
We show that if the success probability as a function of the set of actions is gross substitutes, then an optimal contract can be computed with polynomially many value queries, whereas if it is submodular, the optimal contract is NP-hard.
All our results extend to linear contracts for higher-dimensional outcome spaces, which we show to be robustly optimal given first moment constraints.

Our analysis uncovers a new property of gross substitutes functions, and reveals many interesting connections between combinatorial contracts and combinatorial auctions, where gross substitutes is known to be the frontier for efficient computation.
\end{abstract}

\section{Introduction}

Contract theory is one of the pillars of microeconomic theory (cf., the 2016 Nobel Prize in Economics for Hart and Holmstr\"om~\cite{Nobel16}).  
Indeed, contract theory and its central principal-agent (hidden-action) model play a similar role in markets for \emph{services}, as the theory of mechanism design and its central (combinatorial) auctions model play in markets for \emph{goods}.
The past few years have seen classic applications of contract theory moving online, and with it computational, algorithmic, and optimization approaches are becoming more relevant. Applications range from crowdsourcing platforms, to online labor markets, to online marketing. 

In the classic hidden-action /  principal-agent model of Holmstr\"om~\cite{Holmstrom79} and Grossman and Hart~\cite{GrossmanHart83}, a principal delegates a task to an agent. The agent can take one of $n$ costly actions (e.g., effort levels), and these trigger some distribution over $m$ possible rewards that go to the principal.
The principal cannot observe the action taken by the agent; she can only observe the obtained outcome. 
This model captures an incentive problem, which is quite different from that in mechanism design and auctions, commonly referred to as ``moral hazard'': in and by itself the agent has no interest in taking a costly action. A contract defines a transfer---from the principal to the agent---for each possible stochastic outcome, and serves to incentivize the agent to exert effort. 

In this setting, an \emph{optimal} contract, 
that is, one 
that maximizes the principal's utility assuming that the agent best responds to the contract, can be computed in time polynomial in $n$ and $m$ (as was already shown in~\cite{GrossmanHart83}). 
The idea is to determine for each action whether the principal can incentivize the agent to take it via linear programming. Out of these, the one maximizing the principal's utility will be selected.

An important aspect that the classic model does not explicitly capture---a well recognized fact in the economics literature (cf., the influential multi-tasking paper of Holmstr\"om and Milgrom \cite{HolmstromMilgrom91})---is that typically performing a complex task entails taking a \emph{set} of actions and that the principal's reward as a function of the chosen actions satisfies some form of diminishing returns.

Of course, this could be modeled by the classic model by writing down one meta-action for each of the $2^n$ subsets of actions. However, this approach would ignore all structure of the underlying problem and, in particular, following the above blueprint to compute an optimal contract would 
require considering exponentially many subsets of actions.

\subsection{Our Contribution}
 
In this work we propose a new model of combinatorial contracts that captures the additional structure in problems where a task entails taking a set of actions. 
Within this model we study the design of optimal contracts through a computational lens. 
We establish a non-trivial positive computational result, as well as hardness results, and en-route we reveal several interesting connections to combinatorial auctions.

\paragraph{A model of combinatorial contracts}

In our base model a principal seeks to delegate a project to an agent, and the project can either succeed or fail. We denote the two outcomes with $1$ (for success) and $0$ (for failure). The principal receives a reward of $r(1) \geq 0$ when the project succeeds, and a reward of $r(0) = 0$ otherwise. 
There is a ground set $\actions$ of $n$ actions. The agent can choose any subset $\set \subseteq \actions$ of these actions. If the agent chooses to take actions $\set$, the project succeeds with probability $f(\set)$. Each of the actions $a \in \actions$ has a cost $c(a)$, and the cost of a set of actions $S$ is the sum of the costs $c(a)$ for $a \in S$.

To incentivize the agent to take a certain set of actions, the principal defines a contract. As in the classic model we assume that the actions are hidden, so the contract can only depend on the outcome of the project. Therefore, the principal's only choice is to define payments $t(1)$ and $t(0)$ to the agent for the case that the project succeeds or fails, respectively. The agent will then choose actions so that the expected difference of payment and costs is maximized. The principal's goal in turn is to maximize the expected value minus the payment. 

We also consider a generalization of the base model, in which the outcome space is not binary (a project that succeeds or fails), but rather a vector of $m$ rewards to the principal. In this more general model, each set of actions entails a different probability distribution over rewards, with $f_j(S)$ for $j \in \{0, \ldots, m-1\}$ being the probability of outcome $j$ under actions $S$.

The fact that the principal's reward satisfies some form of diminishing returns (e.g., some form of submodularity) then naturally translates to corresponding assumptions on the probability distributions.

\paragraph{Computing optimal contracts}

Our main results are for the base model where the principal delegates a project to an agent, and the project can be successful or not. While any contract in this binary outcome case is affine, it could pay a positive amount for failure and it could pay more for failure than for success. We observe that this is never optimal, and it suffices to consider contracts that only pay for success. Any such contract can be interpreted as a linear contract, which, for some $\alpha \in [0,1]$, sets $t(1) = \alpha \cdot r(1)$ and $t(0) = 0$. 
Under such a contract, the agent chooses the set $\set$ that maximizes $\alpha \cdot f(\set) \cdot r(1) - \sum_{a \in \set} c(a)$, or equivalently,  $f(\set) \cdot r(1) - \frac{1}{\alpha} \sum_{a \in \set} c(a)$. 
That is, the agent's problem is equivalent to resolving a \emph{demand query} at prices $\frac{1}{\alpha} c(a)$ in the framework of combinatorial auctions.

Our main positive result is for the case where the success probability function $f$ is gross substitutes --- a strict subclass of submodular functions that includes natural functions as special cases (e.g., additive, unit demand, and matroid rank functions), and plays a central role in both economics and computer science. 

\vspace{0.1in}
\noindent {\bf Main Theorem 1:} For gross substitutes success probability functions, the optimal contract can be computed in time polynomial in the number of actions $n$ given access to a value oracle (namely an oracle that for any set $\set$ computes the value $f(\set)$ in polynomial time).
\vspace{0.1in}

A key object in our analysis is the set of {\em critical} values of $\alpha$---values that are potential candidates for an optimal contract.
Our key technical observations are, firstly, that for every function $f$ there are only finitely many critical values of $\alpha$. 
Secondly, in case of a gross-substitutes function $f$, the number of such critical points is bounded polynomially in $n$. Thirdly, we can efficiently iterate over these critical points. Together this gives us a polynomial-time algorithm. 

In order to bound the number of critical points, our main insight is that for a gross-substitutes function $f$ there can be only $O(n^2)$ maximizers $\set$ of $f(\set) \cdot r(1) - \frac{1}{\alpha} \sum_{a \in \set} c(a)$ for different values of $\alpha$. In the language of combinatorial auctions, this means that the number of changes in the demand correspondence when prices (of all items simultaneously) are being scaled linearly is bounded by $O(n^2)$. 
We prove this by uncovering a new property of gross substitutes functions, according to which, generically, as we increase $\alpha$ (decrease prices), whenever the demand set changes, either a single item enters the demand set or an expensive item replaces a cheaper one in a one-to-one fashion. This then implies the claimed bound through a potential function argument. We also show that the upper bound of $O(n^2)$ on the number of critical points for gross-substitute functions $f$ is tight.

For cases beyond gross substitutes success probability functions $f$, we show a hardness result, which applies even for budget additive functions.%

\vspace{0.1in}
\noindent {\bf Main Theorem 2:} For budget additive (and, hence, submodular) success probability functions, computing the optimal contract is NP-hard. 
\vspace{0.1in}

We prove this result by a reduction from the problem of subset sum. In the instance we construct there are only two potential $\alpha$'s that can lead to the optimal contract.
However, to determine which of the two contracts is better,  we need to compute the agent's best response at the corresponding $\alpha$'s, but if we could do this efficiently we could also solve subset sum.

In addition, we show that our approach used for the gross substitutes case, of going over all the critical points, does not work for the case of submodular success probability functions (or even coverage functions). We show this by recursively constructing a coverage function with exponentially many critical points. 

We also show that for every (monotone) success probability function $f$ (whether submodular or not) there is an FPTAS, which computes a $(1-\epsilon)$-approximation with  
$\mathrm{poly}(n,\frac{1}{\epsilon})$ value and demand queries. 
To obtain this result, we observe that we can restrict our attention to $\mathrm{poly}(n,\frac{1}{\epsilon})$ candidate contracts, which we have to compare, one of which will be a $(1-\epsilon)$-approximation.

We then turn to more general outcome spaces, where there is a vector of $m$ rewards, and each set of actions $S \subseteq A$ induces a distribution over these rewards. A linear contract for this setting specifies which fraction $\alpha$ of the reward goes to the agent. Unlike in the case of a binary outcome, linear contracts are no longer optimal for this more general setting. However, as we show, they are robustly optimal in a max-min sense, when for each action only the expected reward rather than the exact distribution is known, and the principal wishes to maximize her utility in the worst-case over all compatible distributions.

Moreover, all our results for optimal (linear) contracts for the case of a binary outcome translate to linear contracts for this more general setting.

\paragraph{Follow-up work}

Deo-Campo Vuong et al.~\cite{DeoCampoVuongEtAl24} and D\"utting et al.~\cite{DuettingFGT24} give an algorithm that finds an optimal contract with access to value and demand oracles, whose running time is polynomial in $n$ and the number of critical points. 
Using this they obtain a poly-time algorithm for supermodular success probability functions $f$.
The work of D\"utting et al.~\cite{DuettingFGT24} also studies a class of matching-based $f$, 
which can be shown to be XOS, a strict superclass of submodular functions. They show that this class of problem instances admits an efficient demand oracle, but can---in general---have a super-polynomial number of critical points. They also present efficient algorithms for several special cases of this problem.

D\"utting et al.~\cite{DuettingFGR24} show that exponentially many value and demand queries are required to find an optimal contract for submodular $f$.
Ezra et al.~\cite{EzraFS24} show that for submodular  $f$, no polynomial time algorithm with value oracle access only can approximate the optimal contract to within any constant factor, unless P = NP. This result applies even if the submodular function is a coverage function. 
For the more general case, where the success probability function $f$ is XOS, Ezra et al.~\cite{EzraFS24} show that, under value oracle access, for any $\epsilon > 0$, no polynomial time algorithm can approximate
the optimal contract to within a factor of $n^{-\frac{1}{2} +\epsilon}$, unless P = NP. Recalling that with value and demand oracle access to $f$ we obtain an FPTAS, this shows the power of demand queries in this setting.

\paragraph{Open problems} 

Despite the progress in follow-up work, several fundamental questions remain open. While Ezra et al.~\cite{EzraFS24} demonstrated that it is impossible to obtain constant-factor approximations for submodular success probabilities in the value oracle model, it would be interesting to identify subclasses of submodular functions that admit such approximation. 
A particularly interesting class of submodular functions, for which it may still be possible to get a PTAS or even an FPTAS, is the class of budget-additive functions. For this class it would also be interesting to prove or disprove whether the number of critical points is polynomially bounded. 
Moreover, while D\"utting et al.~\cite{DuettingFGR24} proved that we can't hope to compute an optimal contract for submodular $f$ even with value and demand oracles, it would be interesting to identify natural settings that admit an efficient demand oracle, for which the problem of computing an optimal contract is hard.

\subsection{Related Work}

\paragraph{Algorithmic contract theory}

The two foundational papers of contract theory are the aforementioned papers of Holmstr\"om \cite{Holmstrom79} and Grossman and Hart \cite{GrossmanHart83}. In addition to the basic model, these papers contain the linear programming approach to computing optimal contracts. In another classic paper, Holmstr\"om and Milgrom \cite{HolmstromMilgrom87} study a multi-round interaction between a principal and an agent, and show that under the assumptions of that model a linear contract is optimal. Holmstr\"om and Milgrom \cite{HolmstromMilgrom91} consider a model similar to ours, but consider a fractional allocation of efforts to actions which makes their model less amenable to a computational analysis. 

Carroll~\cite{10.2307/43495392} assumes that there is a fixed set of known actions but the actual set of actions is a superset of these. He shows that then a linear contract is max-min optimal. In a similar spirit, D\"utting et al.~\cite{DBLP:conf/ec/DuttingRT19} consider the case that only the \emph{expected} rewards of actions are known but not their actual distributions. They show that linear contracts are also max-min optimal in this setting (we show a similar result in Section~\ref{sec:linear-optimal} for the case where the agent chooses any set of actions). 
Besides this, they also discuss how well a linear contract can approximate an optimal one.

Babaioff et al.~\cite{BabaioffFNW12} turn to a setting in which there are multiple agents. They introduce the combinatorial agency model, where a principal incentivizes a \emph{team} of agents to exert costly effort on his behalf in equilibrium, and the outcome depends on the complex combinations of the efforts by the agents. It generalizes an earlier work by Feldman et al.~\cite{FeldmanCSS07} for a simple multi-hop routing. Follow-up work by Babaioff et al.~\cite{BabaioffFN09,BabaioffFN10} studies the effect of mixed strategies and free riding in combinatorial agency. 
More recently, D\"utting et al.~\cite{DuttingEFK23} have explored the multi-agent problem through a similar lens as we do here. Additional results for that model are shown in \cite{DeoCampoVuongEtAl24} and \cite{EzraFS24}.
The combinatorial explosion in these papers comes from a similar source as in our paper. For example, if there are $k$ agents and each agent has $2$ actions, then there are $2^k$ action profiles. An important difference to our work is in the incentive compatibility constraint, which in these papers has to hold for each agent individually, and leads these papers to study equilibria, while in our paper there is only a single agent who will choose the best set of actions. 

D\"utting et al.~\cite{DBLP:journals/siamcomp/DuttingRT21} study a combinatorial contracting problem with one agent choosing one of $n$ actions. There are $m$ different success events, each action causes each event to happen with a certain probability. As these random draws are independent, there are $2^m$ different subsets of success events that can take place, so the number of outcomes is exponential in the input size. This is a similar but orthogonal question to what we study in this paper, where the exponential growth comes from the number of different combinations of actions the agent can choose.

Two additional directions that study contracts through an algorithmic lens, but that are less relevant to this work, are: (1) Work that considers the problem of learning contracts from an (online) no-regret learning perspective (e.g., \cite{HoSV16,DuettingGSW23,ZhuBYWJJ23,DuettingFPS25}), and (2) work that combines moral hazard with screening (e.g., \cite{DBLP:journals/corr/abs-2010-06742,AlonEtAl21,CastiglioniM021,AlonDLT23}).

\paragraph{Incentivizing effort beyond contracts}

The problem of incentivizing effort has also recently been studied in different models. Kleinberg and Kleinberg \cite{DBLP:conf/sigecom/KleinbergK18} and Bechtel et al.~\cite{BechtelD21,BechtelDP22}, for instance, study problems of algorithmic delegation. In another related direction, Li et al.~\cite{LiHSW22} and Hartline et al.~\cite{HartlineSLW23} explore the use of scoring rules for incentivizing effort in information elicitation problems. Finally, Kleinberg and Raghavan \cite{DBLP:conf/ec/KleinbergR19} (and follow-ups) consider problems of strategic classification.

\paragraph{Gross substitutes functions}

Gross substitute functions play a central role in economics (e.g., \cite{kelso1982job}). They have been independently discovered in mathematics, under a different name in the context of discrete convex analysis, see \cite{MurotaS99,Murota96}.
The class of gross substitutes functions is a strict subclass of submodular functions \cite{LehmannLN06}, which includes natural functions such as additive, unit demand, and matroid rank functions as special cases. 
This class plays a central role in the analysis of combinatorial markets; for example, it is the frontier for both market equilibrium existence \cite{kelso1982job,GS99}, and for the efficient computation of a welfare-maximizing allocation \cite{nisan2006communication}. 
Its original definition uses the notions of prices, utility and demand \cite{kelso1982job}, but due to its centrality in combinatorial markets, it has attracted a lot of work that furthers our understanding of its characteristics (e.g., \cite{BalkanskiL18,Murota2016,DobzinskiFF21}); see \cite{PaesLeme17} for an influential algorithmic survey.

\paragraph{Assortment optimization} 

There is also a surprising connection between our perspective on contract theory and assortment optimization as introduced by Talluri and van Ryzen \cite{DBLP:journals/mansci/TalluriR04}. For our main positive result, we show that for a gross-substitutes function $f$ there can be only $O(n^2)$ maximizers $\set$ of $f(\set) - \frac{1}{\alpha} \sum_{a \in \set} c(a)$ for different values of $\alpha$. In the language of Talluri and van Ryzen \cite{DBLP:journals/mansci/TalluriR04}, this means that if the total probability of purchase is a gross-substitutes function, then the number of efficient sets is in $O(n^2)$, making their Bellman equations efficiently solvable. The work by Immorlica et al.~\cite{DBLP:journals/teco/ImmorlicaLMST21} uses similar perspectives, which could possibly also be made (more) efficient by our insights.

\subsection{Paper Structure}

We present our problem in Section~\ref{sec:model}. In Section~\ref{sec:insights} we present some useful insights on the structure of the problem.
Our polynomial time algorithm for gross substitutes success probability functions is presented in  Section~\ref{sec:gs}. 
In Section~\ref{sec:beyond} we study success probability functions beyond gross substitutes: in Section~\ref{sec:submodular} we provide negative results for submodular functions; in Section~\ref{sec:ptas} we present a FPTAS for general success probability functions; and in Section~\ref{sec:weakly-poly} we present a weakly poly-time algorithm for instances with poly-size critical sets.
In Section~\ref{sec:linear-optimal}, we discuss the case of non-binary outcomes.
Omitted proofs are deferred to the appendix.

\section{The Combinatorial Contracts Problem}
\label{sec:model}

\paragraph{Hidden-action principal-agent setting} 

There is a single principal and a single agent. There is a set $\actions = \{1, \ldots, n\}$ of $n$ possible actions. The strategy of the agent consists of a set of actions $\set \in \actionset$. Every action $a \in \actions$ is associated with a positive cost $c(a)>0$. The cost of a set of actions $\set \in \actionset$ is additive; i.e., $c(\set)=\sum_{a \in \set}c(a)$. The cost of not taking any actions $c(\emptyset)$ is zero.

We focus on the case of a binary outcome space $\Omega = \{0,1\}$. (We consider more general, higher dimensional outcome spaces in Section~\ref{sec:linear-optimal}.) Outcome $0$ corresponds to failure, and outcome $1$ corresponds to success. The principal derives a reward $r(1) \in \mathbb{R}_{\geq 0}$ from outcome $1$ (i.e., success), and a reward of $r(0) = 0$ from outcome $0$ (i.e., failure).

Every strategy $\set \in \actionset$ by the agent has an associated success probability $f: \actionset \rightarrow [0,1]$. We assume that $f(\emptyset) = 0$, and that $f$ is monotonically non-decreasing, i.e., $\set \subseteq \set'$  implies that $f(\set) \leq f(\set')$.

The principal cannot directly observe the set of actions chosen by the agent, but she can observe the stochastic outcome of the chosen set of actions.

\paragraph{The contract design problem}

A contract $t: \Omega \rightarrow \mathbb{R}_{\geq 0}$ is a mapping from outcomes to non-negative payments for each outcome. In the binary case, a contract thus corresponds to two numbers, $t(0)$ and $t(1)$, the payment upon failure and success.

The principal's expected reward for a set of actions $\set \in \actionset$ is $R(\set) = f(\set) \cdot r(1)$. The expected payment from the principal to the agent for a set of actions $\set \in \actionset$ is defined as $T(\set) = (1-f(\set)) \cdot t(0) + f(\set) \cdot t(1)$. The principal's expected utility from a set of actions $\set \in \actionset$ is
\[
    u_p(\set,t) = R(\set) - T(\set).
\]
The agent's expected utility from a set of actions $S \in \actionset$ is 
\[
    u_a(\set,t) = T(\set) - c(\set).
\]

A set of actions $S \in \actionset$ is \emph{a best response} to a contract $t$ (we also say it is \emph{incentivized} by contract $t$) if it yields the highest possible utility to the agent, where we assume that the agent breaks ties in favor of the principal. 

Formally, let 
\[
    \demandset(t) = \argmax_{\set' \in \actionset} u_a(\set',t)
\]
be the collection of sets of actions that maximize the agent's utility. Then the collection of sets of actions that are incentivized by the contract is
\[
    \demandset^\star(t)= \argmax_{\set' \in \demandset(t)} u_p(\set',t) \subseteq \demandset(t).
\]

The assumption that the agent breaks ties in favor of the principal is a standard one in the contracts literature (see, e.g., \cite{10.2307/43495392}). It is motivated by the fact that one could perturb payments slightly to achieve the same effect. 

Note that we have set things up so that the best response condition (or the ``IC constraint'') implies individual rationality, i.e., that the agent's utility is non-negative. 
Also note that for any $\set \in \demandset^\star(t)$, the principal's utility $u_p(\set,t)$ is the same. We can thus define $u_p(t)$ to be the principal's utility $u_p(\set,t)$ from any set $\set \in \demandset^\star(t)$. 

The computational problem that we are interested in is that of computing a contract that maximizes the principal's utility.

\bigskip

\begin{tabular}{p{1.5cm}p{8.6cm}}
\multicolumn{2}{l}{\textsf{OPT-CONTRACT}:}\\
{\bf Input:} & Action set $\actions$, outcome set $\Omega$, rewards $r(j)$ for $j \in \Omega$, costs $c(i)$ for $i \in A$, oracle access to $f$\\
{\bf Output:} & Contract $t$ that maximizes $u_p(t)$
\end{tabular}

\bigskip

The following simple observation, will allow us to narrow down the search space: 

\begin{restatable}{observation}{propzerozero}
    For any contract $t$ there is a contract $t'$ such that $t'(0) = 0$ that yields a weakly higher utility to the principal.
\end{restatable}

We can thus focus on contracts $t$ such that $t(0) = 0$. Any such contract can be expressed by a single parameter $\alpha$ such that $t(1) = \alpha \cdot r(1)$. This motivates identifying contracts $t$ with their $\alpha$, and replacing $t$ with $\alpha$ and $t(1)$ with $\alpha \cdot r(1)$ in the definitions above. 
For example, we replace  $\demandset(t)$ and $\demandset^{\star}(t)$ by $\demandset(\alpha)$ and $\demandset^{\star}(\alpha)$, respectively.

Hereafter, we normalize $r(1)$ to be $1$. 
Hence, $R(\set) = f(\set)$, $t(1)=\alpha$, and $f(\set)\cdot r(1)$ can be replaced by $f(\set)$.

\begin{example} \label{ex:demand} 
    Consider the success probability function $f: 2^{\{1,2,3\}} \rightarrow [0,1]$ where:
    $f(\emptyset)=0,~f(\{1\}) = f(\{2\}) = 0.35, f(\{1,2\})=0.5,~f(\{3\}) =f(\{1,3\}) =f(\{2,3\})=f(\{1,2,3\})=0.6,$
    and $c(1)=c(2)=0.05,~ c(3)=0.15$.
    (One can verify that $f$ is submodular but not gross substitutes, see definitions below.)
    Consider two contracts, $\alpha_1=0.5$ and $\alpha_2=0.25$.
    For $\alpha_1=0.5$ it holds that $\demandset(\alpha_1) = \{\{1,2\},\{3\}\}$ since both sets give the agent an expected utility of $0.15$. On the other hand $\demandset^\star(\alpha_1)=\{\{3\}\}$ since this is the only set in the demand that maximizes the utility of the principal.
    For $\alpha_2=0.25$ it holds that $\demandset(\alpha_2)=\demandset^\star(\alpha_2) = \{\{1\},\{2\}\}$ since both sets
    maximize the expected utility of the agent, and give the same principal's utility.
\end{example}

\paragraph{The agent's problem} 

For a fixed contract $t$ with $t(0) = 0$ or the corresponding $\alpha$, 
the agent's problem is to find a set $\set \in \demandset^\star(\alpha)$.

\bigskip

\begin{tabular}{p{1.5cm}p{8.6cm}}
\multicolumn{2}{l}{\textsf{BEST-RESPONSE}:}\\
{\bf Input:} &Contract $\alpha \in (0,1]$\\
{\bf Output:} &Some set $S \in \demandset^\star(\alpha)$
\end{tabular}

\bigskip

To determine whether $\set \in \demandset(\alpha)$ for some $\alpha \in (0,1]$ we need to compare the agent's utility for pairs of sets of actions $\set, \set' \in \actionset$. Rather than comparing $u(\set,\alpha)$ to $u(\set',\alpha)$ we can equivalently compare $u(\set,\alpha)/\alpha$ and $u(\set',\alpha)/\alpha$, i.e., 
\[
    f(\set)  - \sum_{i \in \set} c(i)/\alpha \quad \text{and}\quad f(\set') - \sum_{i \in \set'} c(i)/\alpha.
\]

\paragraph{Success probability functions} 

It is natural to impose some form of ``decreasing marginal returns'' on the success probabilities $f: \actionset \rightarrow [0,1]$ as a function of the set of actions taken. We consider the following classes of functions:

\begin{itemize}
    \item The function $f$ is \emph{submodular} if for every $\set, \set' \in \actionset$ with $\set \subseteq \set'$ and every  $i \in \actions \setminus \set'$ we have $f(\set \cup \{i\}) - f(\set) \geq f(\set' \cup \{i\}) - f(\set')$. 
	
    \item The function $f$ is \emph{budget additive} if there exists a budget $B \in [0,1]$ such that for every $\set \in \actionset$ we have $f(\set) = \min\{B, \sum_{i \in \set} f(\{i\})\}$.
		
    \item The function $f$ is {\em coverage} if there exists a finite set $U$, where every element $j \in U$ is associated with a weight $w_j \in \reals_{\geq 0}$, and a function $g:\actions \rightarrow 2^U$ such that for every set $S \in 2^\actions$, $f(S)=\sum_{j \in \bigcup_{i\in S}g(i)}w_j$.
	
    \item The function $f$ is \emph{gross substitutes} if for any two vectors $p, q \in \reals_{\geq 0}^n$ such that $q \geq p$ and any $\set \in \actionset$ such that $f(\set) - \sum_{i \in \set} p_i \in \argmax_{\set' \in \actionset} f(\set') - \sum_{i \in \set'} p_i$ there is a $\altset \in \actionset$ such that $f(\altset) - \sum_{i \in \altset} q_i \in \argmax_{\altset' \in \actionset} f(\altset') - \sum_{i \in \altset'} q_i$ and $\altset \supseteq \{i \in \set \mid q_i = p_i\}.$ 
    I.e., whenever the prices of some items increase and the prices of other items remain intact, the agent's demand for the items whose price remain intact weakly increases.
	
    \item The function $f$ is \emph{unit demand} if for every $\set \in \actionset$, 
    $f(\set) = \max_{i \in \set} f(\{i\})$.
\end{itemize}

Unit demand functions are gross substitutes. Gross substitutes functions, coverage functions, and budget additive functions are submodular. Gross substitutes functions, coverage functions, and budget additive functions are incomparable to each other.

\paragraph{Value vs.~demand queries} 

The success probabilities $f: \actionset \rightarrow [0,1]$ are combinatorial objects, whose explicit description size is exponential in $n$. For computational questions we therefore consider the following two types of oracle access to these functions:
\begin{itemize} 
    \item A \emph{value oracle} receives a set $\set \in \actionset$ as input and returns $f(\set)$.
    \item A \emph{demand oracle} receives a vector $p \in \reals^n_{\geq 0}$ of ``prices'' as input and returns some set $\set$ that maximizes the ``utility'' $f(\set) - \sum_{i \in \set} p_i$.
\end{itemize}

Demand oracles are a natural assumption in the context of combinatorial auctions, where they correspond to asking the agent for a set of items that maximizes his utility given item prices. 
As we shall see, they also play a natural role in combinatorial contracts.
For example, the agent's problem above is essentially solving a demand query.

\section{Structural Insights}
\label{sec:insights}

In this section we present some useful insights regarding the structure of $\demandset_{f,c}(\alpha)$ and $\demandset_{f,c}^\star(\alpha)$ (recall that $\demandset_{f,c}(\alpha)$ is the collection of sets of actions maximizing the agent's utility under contract $\alpha$, and $\demandset_{f,c}^\star(\alpha)$ is the subset among these that maximize the principal's utility). 
We also define {\em critical} values of $\alpha$ as ones for which the demand changes (formal details below). 
As we shall see, the set of critical values of $\alpha$ is useful in calculating the optimal contract.

Given some $\alpha$, let $\g_{f,c}(\alpha) = \max_{\set\in \demandset_{f,c}(\alpha)} f(\set)$. Note that for every $\alpha<1$ and every $\set\in  \demandset^\star_{f,c}(\alpha)$ it holds that $\g_{f,c}(\alpha) = f(\set)$.  

For example, in Example~\ref{ex:demand}, it holds that 
$\g_{f,c}(0.5) = 0.6 $ since $\demandset^\star_{f,c}(0.5)=\{\{3\}\}$ and $f(\{3\})=0.6$. Similarly, $\g_{f,c}(0.25) = 0.35$ since $\demandset^\star_{f,c}(0.25)=\{\{1\},\{2\}\}$ and $f(\{1\})=f(\{2\})=0.35$.

The following proposition establishes a monotonicity property of the demand as a function of $\alpha$.
\begin{proposition}\label{prop:monotone}
    Let $0\leq \alpha_1 <\alpha_2 $.
    For every
    $\set_1\in\demandset_{f,c}(\alpha_1)$ and $\set_2\in\demandset_{f,c}(\alpha_2)$, it holds that $f(\set_1) \leq f(\set_2)$.
\end{proposition}
\begin{proof}
    Let $\set_1\in\demandset_{f,c}(\alpha_1)$ and let $\set_2\in\demandset_{f,c}(\alpha_2)$. By these definitions, we have $\alpha_1 f(\set_1) - c(\set_1) \geq \alpha_1 f(\set_2) - c(\set_2)$ and $\alpha_2 f(\set_2) - c(\set_2) \geq \alpha_2 f(\set_1) - c(\set_1)$. Adding these inequalities implies $\alpha_1 f(\set_1) + \alpha_2 f(\set_2) \geq \alpha_2 f(\set_1) + \alpha_1 f(\set_2)$, or equivalently, $(\alpha_2 - \alpha_1) f(\set_1) \leq (\alpha_2 - \alpha_1) f(\set_2)$. Dividing by $\alpha_2 - \alpha_1$ implies the claim.
\end{proof}

Note that $\g_{f,c}(\alpha) = \max_{\set \in \demandset_{f,c}(\alpha)} f(\set)$. Thus, Proposition~\ref{prop:monotone} implies monotonicity of $\g_{f, c}$.
\begin{corollary}
    \label{cor:mon-g}
    For every $0 \leq \alpha_1 < \alpha_2$, $\g_{f,c}(\alpha_1) \leq \g_{f,c}(\alpha_2)$.
\end{corollary}

We next observe that $\g_{f,c}$ is right continuous. The proof  exploits the fact that the agent breaks ties in favor of the principal. 

\begin{restatable}{observation}{obscontright}
    For every $\alpha>0$, 
    $\lim_{\epsilon\rightarrow 0^+}\g_{f,c}(\alpha + \epsilon)$ and $\lim_{\epsilon\rightarrow 0^+}\g_{f,c}(\alpha - \epsilon)$ are well defined. 
    Moreover, $$\g_{f,c}(\alpha) = \lim_{\epsilon\rightarrow 0^+}\g_{f,c}(\alpha + \epsilon).$$
    \label{obs:properties}
\end{restatable}

To summarize, $\g_{f,c}$ is a monotone, right-continuous function whose image is contained in the image of $f$. Thus, it must be a ``step function'' with at most $2^n$ steps. This is cast in the following corollary.

\begin{corollary}\label{cor:structure}
    There exists some $k < 2^n$ and a series of $\alpha$ values $0=\alpha_0<\alpha_1 < \ldots < \alpha_k \leq  1$ such that for every  $x \in [0,1]$, $\g_{f,c}(x) = \g_{f,c}(\alpha)$, where $\alpha = \max\{\alpha_i \mid \alpha_i \leq x\}$. 
    Moreover, for every $0 \leq i < j \leq k$,  $\g_{f,c}(\alpha_i) < \g_{f,c}(\alpha_j)$.  
\end{corollary}

Every $\alpha_i$ from the last corollary (except for $\alpha_0$) is said to be a {\em critical} value of $\alpha$. The set $\{\alpha_1, \ldots, \alpha_k\}$ is termed the {\em critical set} with respect to $f,c$, and is denoted by $\orbit{f,c}$; i.e., 
\begin{equation*}
    \orbit{f,c} = \{\alpha \in (0,1] \mid \g_{f,c}(\alpha) \neq \g_{f,c}(\alpha - \epsilon) \quad \forall \epsilon > 0\}. 
\end{equation*}

The principal's utility is $(1-\alpha)\cdot \g_{f,c}(\alpha)$. By Corollary~\ref{cor:structure}, $\g_{f,c}(\alpha)$ is constant for every $\alpha$ that lies between two consecutive critical $\alpha$'s. This implies that the principal can restrict the search space to $\alpha$ values in the critical set $\orbit{f,c}$ (or zero); that is:

\begin{observation} \label{obs:opt}
    Let $\alpha^\star$ be the optimal contract with respect to $f,c$. Then, $\alpha^\star \in \orbit{f,c} \cup \{0\}$.
\end{observation}

For example, suppose $f$ is an additive success probability function. 
Then, for every contract $\alpha$, the best response of the agent is to select all actions $a\in \actions$ such that $\alpha \cdot f(a) \geq c(a)$.
Consequently, the critical set $\orbit{f,c}$ 
is $\{\frac{c(a)}{f(a)} \mid a \in \actions \}$.
Thus, the optimal contract $\alpha^\star$ is the best point among these critical $\alpha$ values.
As we shall see, things become more complex for richer classes of success probability functions.

\section{Gross Substitutes Functions} 
\label{sec:gs}

In this section, we devise a polynomial time algorithm for the case that the success probability function $f$ is gross-substitutes, which includes additive, unit demand, and matroid rank functions as special cases.

The two main insights that drive this result are as follows. First---as formalized in Section~\ref{sec:implementation}, with the key insight appearing in Lemma~\ref{lem:fsa}---we can exploit the greedy nature of the demand oracle problem for gross-substitutes $f$ to obtain a procedure that efficiently finds the next higher critical value of $\alpha.$ With this we can find an optimal contract by iterating over all critical values, provided that the number of critical values is polynomially bounded. Second---as formalized in Section~\ref{sec:critical-set}, with the key insight appearing in Lemma~\ref{lemma:ispotential}---we show that for gross-substitutes $f$, as we change $\alpha$, the demand set can only change in very specific ways. Namely, whenever the demand set changes, 
either a single action enters the demand set or an expensive action replaces a cheaper
one in a one-to-one fashion. This, in turn, enables a potential argument that lets us bound the number of critical values by $O(n^2)$.

\begin{theorem}\label{thm:gs-poly}
    Given a gross substitutes success probability function $f$, and an additive cost function $c$, one can compute the optimal contract $\alpha$ in polynomial time.
\end{theorem}

To prove Theorem~\ref{thm:gs-poly}, we present Algorithm~\ref{alg:contract}, which computes the optimal contract for any success probability function $f$, and show that it can be implemented in polynomial time for the case where $f$ is gross substitutes. 

Algorithm~\ref{alg:contract} uses the {\em successor} function, which, for any value in $[0,1]$, returns the next critical $\alpha$. Formally, the successor is a function $\successor_{f,c}: [0,1] \rightarrow [0,1]\cup \{\nill\}$ which, for any value of $\alpha \in [0,1]$, returns the smallest $\alpha' > \alpha$ such that $\alpha' \in \orbit{f,c}$ or $\nill$ if such an $\alpha'$ does not exist.

\begin{algorithm}
\caption{Optimal contract}
\label{alg:contract}
\begin{algorithmic}
\STATE{\textbf{Result:} Optimal contract $\alpha$}
\STATE{\textbf{Input:} Success probability function $f$, and costs $c$}
\STATE $\alpha^\star = \alpha^{(0)} = 0$
\STATE $t=1$
\STATE $\alpha^{(1)}=\successor_{f,c}(\alpha^{(0)})$ 
\WHILE{$\alpha^{(t)} \neq \nill$} 
\IF{$(1-\alpha^{(t)} ) \cdot \g_{f,c}(\alpha^{(t)}) > (1-\alpha^\star ) \cdot \g_{f,c}(\alpha^\star)$} 
\STATE $\alpha^\star= \alpha^{(t)}$
\ENDIF	
\STATE $t=t+1$
\STATE $\alpha^{(t)} =\successor_{f,c}(\alpha^{(t-1)})$
\ENDWHILE
\RETURN $\alpha^\star$
\end{algorithmic}
\end{algorithm}

Algorithm~\ref{alg:contract} is a generic algorithm for finding the optimal contract. The algorithm goes over all critical  $\alpha$'s and returns the best one among them (which by Observation~\ref{obs:opt} is the optimal contract). 
It assumes the existence of $\successor_{f,c}$ and $\g_{f,c}$ oracles. Its complexity is bounded by the size of the critical set $\orbit{f,c}$ and the complexity required for computing  $\successor_{f_c}$ and $\g_{f,c}$.

\begin{theorem}
    \label{th:poly-opt-contract}
    Algorithm~\ref{alg:contract} returns the optimal contract. Its time complexity is bounded by $|\orbit{f,c} |$ multiplied by the complexity of computing $\successor_{f,c}$ and $\g_{f,c}$.
\end{theorem}

In the remainder of this section we show that for every gross substitutes success probability function $f$, Algorithm~\ref{alg:contract} runs in polynomial time. 
In Section~\ref{sec:implementation} we show that $\successor_{f,c}$ and $\g_{f,c}$ can be implemented in polynomial time, and in Section~\ref{sec:critical-set} we show that $|\orbit{f,c}| \leq \frac{n (n+1)}{2}$.

\subsection{\texorpdfstring{Implementation of $\successor_{f,c}$ and $\g_{f,c}$}{Implementation of succ and V}}
\label{sec:implementation}

It is well known that a demand query for gross substitutes functions can be computed in polynomial time in $n$ by the following greedy algorithm: As long as there is an action with non-negative marginal utility (added value minus cost of the action), pick an action with maximal marginal utility (see, e.g., \cite{PaesLeme17}). Generally, ties can be broken arbitrarily in this greedy procedure. For our purposes, it will be helpful to consider a particular tie-breaking rule, as follows.

\begin{definition}
    \label{def:greedy}
    The ordered demanded set with respect to a contract $\alpha$, denoted $\setal \in \demandset_{f,c}(\alpha)$, is the ordered set obtained by the greedy algorithm with the following tie-breaking rule: \begin{itemize}
        \item Among multiple actions with the same (highest) marginal utility, pick an action with maximal cost. Among these, pick the action with the smallest index. 
        \item If the highest marginal utility is 0, pick such an action. 
    \end{itemize} 
\end{definition}

We refer to the greedy algorithm with the  tie-breaking rule in Definition~\ref{def:greedy} as \greedy. Let $\setal[t]$ denote the action selected by \greedy\ in the $t$-th iteration, and let $\setal^t$ denote the set of actions selected within the first $t$ iterations, that is, $\setal^t = \{\setal[1],\ldots,\setal[t]\}$.
 
We first claim that the set $\setal$  maximizes the principal's utility among all sets in the agent's demand.
\begin{proposition}
    For every $\alpha$, it holds that $\setal \in \demandset^\star_{f,c}(\alpha)$.
\end{proposition}

\begin{proof}
    Consider a contract $\alpha' =\alpha+\epsilon$ for a sufficiently small $\epsilon$. 
    Let $B_t$ be the set of actions for which \greedy~breaks ties at step $t$ (with respect to $\alpha$),  i.e., $$B_t=\left\{a\in \actions \setminus  \setal^{t-1} \mid f(a\mid \setal^{t-1}) -\frac{c(a)}{\alpha} = f(\setal[t]\mid \setal^{t-1}) -\frac{c(\setal[t])}{\alpha}\right\}.$$ 
    In what follows we
    prove by induction on $t$, that for sufficiently small $\epsilon>0$, \greedy~on $\alpha'$ returns $S_\alpha$. 
    Since $\epsilon>0$, by  Definition~\ref{def:greedy}, when  \greedy~runs on $\alpha'>\alpha$, we get that among actions in $B_t$, it selects at round $t$ action $\setal[t]$. This is since the increase in the marginal value of an action $a\in B_t$ from $\alpha$ to $\alpha'$ is $\frac{c(a)}{\alpha} -\frac{c(a)}{\alpha'}$, which is strictly monotone in $c(a)$, and the tie-breaking rule chooses the one with the highest cost in a consistent way.
    We know that in \greedy~on $\alpha$ at time $t$, every action $a$ in  $B_t$ is strictly preferred over every action $a'$ in $(A \setminus \setal^{t-1}) \setminus B^t$, i.e., 
    $$ 
        f(a\mid \setal^{t-1}) -\frac{c(a)}{\alpha } > f(a'\mid \setal^{t-1}) -\frac{c(a')}{\alpha }.
    $$ 
    Thus, by continuity of the marginal utility of actions in $\alpha$, we get that for small enough value of $\epsilon>0$, all actions in  $B_t$ are strictly preferred over all actions in  $(A \setminus \setal^{t-1}) \setminus B^t$ with respect to $\alpha'$.
    Thus, at each step $t$, \greedy~on $\alpha'$ selects action $\setal[t]$, which implies that
    $S_\alpha = S_{\alpha'}$. 
    Combining with Proposition~\ref{prop:monotone}, this shows that $f(\setal) \geq f(\set)$ for every $\set \in \demandset_{f,c}(\alpha)$, implying that $\setal \in \demandset^\star_{f,c}(\alpha)$.
\end{proof}

We next show that $\nextf(\alpha)$ can be computed in polynomial time.

\begin{lemma}\label{lem:fsa}
    Given a gross substitutes function $f$, an additive cost function $c$, and $\alpha \geq 0 $, one can compute $\nextf(\alpha)$ in polynomial time.
\end{lemma}
\begin{proof}
    Consider the case where $\nextf(\alpha)\neq \nill$. Let $\alpha' = \nextf(\alpha)$, let $\setal$ and $\setalp$ be the respective sets returned by \greedy, and let $d=|\setal|$. Observe that there has to be an $i \in [d]$ such that $\setal[i] \neq \setalp[i]$ or that $\lvert \setalp \rvert > d$. In the former case, consider the smallest such $i$. We have $\setal^{i-1} = \setalp^{i-1}$ and $(\alpha' - \epsilon) f(\setal[i] \mid \setal^{i-1}) - c(\setal[i]) \geq (\alpha' - \epsilon) f(\setalp[i] \mid \setal^{i-1}) - c(\setalp[i])$ for all sufficiently small $\epsilon > 0$ as well as $\alpha' f(\setalp[i] \mid \setalp^{i-1}) - c(\setalp[i]) \geq \alpha' f(\setal[i] \mid \setalp^{i-1}) - c(\setal[i])$. This implies that $\alpha' = \frac{c(\setalp[i]) - c(\setal[i])}{f(\setalp[i] \mid \setal^{i-1}) - f(\setal[i] \mid \setal^{i-1})}$. Analogously, if $\lvert \setalp \rvert > d$, then $\alpha' = \frac{c(\setalp[d+1])}{f(\setalp[d+1] \mid \setal)}$.
	
    So, in order to compute $\nextf(\alpha)$, it suffices to consider all (finite) ratios 
    \begin{align*}
    \frac{c(a) - c(\setal[i])}{f(a \mid \setal^{i-1}) - f(\setal[i] \mid \setal^{i-1})} \quad \text{and} \quad \frac{c(a)}{f(a \mid \setal)}
    \end{align*}
    for all $a \in \actions$ and $i \in [d]$. The smallest one that is bigger than $\alpha$ and has a larger value of $\g_{f,c}$ is $\nextf(\alpha)$. If there is none, then $\nextf(\alpha)=\nill$. Since there are at most $n^2$ such ratios, this can be implemented in polynomial time.
\end{proof}

\subsection{Bounding the Size of the Critical Set}
\label{sec:critical-set}

In this section we establish an upper bound on the size of the critical set, for every gross substitutes function and every additive cost function.
A matching lower bound appears in Proposition~\ref{prop:gs-lower-bound} in the appendix.

\begin{theorem}
\label{theorem:gs}
    Let $f$ be a gross substitutes function, and let $c$ be an additive cost function. 
    It holds that $|\orbit{f,c}| \leq \frac{n (n+1)}{2}$. 
\end{theorem}

To prove the theorem, we proceed as follows. In Section~\ref{sec:generic-cost}, we introduce the notion of generic cost functions and prove the theorem for this subclass. Specifically, a generic cost function ensures that for every contract $\alpha>0$, the \greedy\ algorithm encounters at most one tie between two actions. Under this assumption, we characterize how the agent's best response changes as $\alpha$ increases. 
In Section~\ref{sec:arbitrary-cost}, we show that although the structure of best response changes might be different under arbitrary cost functions, the analysis of generic cost functions still bounds the number of critical points.
To this end, we employ a perturbation argument, showing that generalizing beyond generic cost functions can only reduce the number of critical points.

\subsubsection{Generic Cost Functions}
\label{sec:generic-cost}

We wish to define the notion of a {\em generic} cost function with respect to a success probability function so that for every $\alpha$, \greedy\ has at most one round where tie breaking occurs.
In order to formally define this notion, we first define sets of candidate critical values.

\begin{definition} \label{def:candidates}
    Given a function $f$, a cost function $c$, and an (unordered) pair of actions $a_1,a_2\in \actions \cup \{\nill\}$ such that $a_1 \neq a_2$, let 
    \begin{align*}
        \cand{f,c}{a_1,a_2}=\{\alpha \mid        \exists_{\set_1,\set_2 \subseteq \actions}~\alpha f(a_1 \mid \set_1) - c(a_1) = 
	\alpha f(a_2 \mid \set_2) - c(a_2) \},
    \end{align*}
    where $c(\nill)=0$ and $f(\nill \mid S)=0$ for every set $S$.
\end{definition}

That is, $\cand{f,c}{a_1,a_2}$ is the set of values of $\alpha$ such that the marginal utility of $a_1$ with respect to some set $\set_1$ equals the marginal utility of $a_2$ with respect to some set $\set_2$.
These are candidate values of $\alpha$ for which \greedy\ may be indifferent between adding $a_1$ and adding $a_2$, where $a_1$ (or $a_2$) may be $\nill$.

We next observe that only $\alpha$ values that belong to some $\cand{f,c}{a_1,a_2}$ may be critical.

\begin{observation} \label{obs:must-cand}
    If for every $a_1,a_2\in \actions \cup \{\nill\}$ such that $a_1 \neq a_2$ it holds that $\alpha \not\in \cand{f,c}{a_1,a_2}$, then $\alpha \not\in \orbit{f,c}$.		
\end{observation}

\begin{proof}
    Since $\alpha \not\in \cand{f,c}{a_1,a_2}$ for any $a_1,a_2\in \actions \cup \{\nill\}$, $a_1 \neq a_2$, the execution of \greedy\ does not involve any tie breaking because always $\alpha f(\setal[t] \mid \set^{t-1}) - c(\setal[t]) > \alpha f(a \mid \set^{t-1}) - c(a)$ for all $a \neq \setal[t]$. As all of these finitely many inequalities are strict, there has to be an $\epsilon > 0$ such that still $(\alpha - \epsilon) f(\setal[t] \mid \set^{t-1}) - c(\setal[t]) > (\alpha - \epsilon) f(a \mid \set^{t-1}) - c(a)$ for all $t$ and $a$. That is, on $\alpha - \epsilon$, the execution of \greedy\ is identical to the one on $\alpha$. So, $\alpha$ cannot be critical.
\end{proof}

We are now ready to define generic cost functions.

\begin{definition}\label{def:generic-prices}
    An additive cost function $c$ is said to be \textit{generic} w.r.t. a success probability function $f$ if for every $\alpha >0$, there exists at most one (unordered) pair of actions $a_1,a_2 \in \actions \cup \{\nill\}$, $a_1 \neq a_2$, such that $\alpha \in \cand{f,c}{a_1,a_2}$.
\end{definition}

By definition, for every generic cost function $c$, and every $\alpha>0$, there could be at most one iteration in \greedy\ (i.e., when generating $\setal$) in which tie breaking occurs.

The following lemma establishes an upper bound on the size of the critical set for every generic cost function.

\begin{proposition}\label{lem:oracle-generic}
    Let $f$ be a gross substitutes function, and let $c$ be a generic cost function w.r.t. $f$. 
    It holds that $|\orbit{f,c}| \leq \frac{n (n+1)}{2}$. 
\end{proposition}
	
We will prove the claim using a potential function $\poten\colon 2^{\actions} \rightarrow \mathbb{Z}_{\geq 0}$, which associates every set of actions $\set$ with a non-negative integer. In order to define $\poten$, for every action $a \in \actions$, we let $r_a$ be the rank of $a$ according to $c$ (i.e., the rank of the highest cost action is $n$, the rank of the 2nd highest cost action is $n-1$, and so on).
		
Since the cost function is assumed to be generic, the rank is unique. To see this, observe that for every generic cost function it holds that for any two actions $a_1 \neq a_2$ it must be that $c(a_1) \neq c(a_2)$, or else $\cand{f,c}{a_1,a_2}$ would be the set of all real numbers (by setting $\set_1=\set_2=\actions$).
	
Now, we define $\poten(\set)=\sum_{a\in \set} r_a$.
The core insight is the following lemma, showing that the potential increases in $\alpha$. 

\begin{lemma}\label{lemma:ispotential}
    For every $\alpha', \alpha \in \orbit{f,c}$, $\alpha' < \alpha$, we have $\poten(\setalp) \leq \poten(\setal)-1$.
\end{lemma}
	
Before proving Lemma~\ref{lemma:ispotential}, we show how it implies Proposition~\ref{lem:oracle-generic}.
Letting $\orbit{f,c} = \{ \alpha_1, \ldots, \alpha_k \}$ with $\alpha_1 < \ldots < \alpha_k$, we have $\poten(\set_{\alpha_1}) \geq 1$ because $\set_{\alpha_1} \neq \emptyset$. 
We get that $\poten(\set_{\alpha_{j+1}}) \geq \poten(\set_{\alpha_j}) +1$ for all $j$, implying that $\poten(\set_{\alpha_k}) \geq k$. However, we also have, $\poten(\set_{\alpha_k}) \leq \sum_{a \in \actions} r_a = \sum_{i=1}^n i = \frac{n(n+1)}{2}$. This implies that $k \leq \frac{n(n+1)}{2}$.
	
We now prove Lemma~\ref{lemma:ispotential}.

\begin{proof}[Proof of Lemma~\ref{lemma:ispotential}]
    In order to prove Lemma~\ref{lemma:ispotential}, it suffices to prove that 
    $\poten(\setalp) < \poten(\setal)$ for every {\em neighboring} $\alpha, \alpha' \in \orbit{f,c}$. That is, $\alpha, \alpha'$ such that $\alpha' < \alpha$ and $(\alpha', \alpha) \cap \orbit{f,c} = \emptyset$.
	
    Specifically, it suffices to prove that for any neighboring $\alpha'$ and $\alpha$, the set  $\setalp$ takes one of the following two forms: either (i)  $ \setalp = \setal\setminus \{a\}$ for some $a \in \setal$, or (ii) $\setalp = (\setal\setminus \{a_1\}) \cup \{a_2\}$ for some $a_1 \in \setal$, and $a_2\notin \setal$, where $c(a_2) < c(a_1)$.
    Indeed, in each one of these cases, the potential of $\setalp$ is smaller than the potential of $\setal$ by at least 1.
    
    We note that since $\alpha, \alpha'$ are neighboring, it holds that $\set_{\alpha'}=\set_{\alpha-\epsilon}$ for every $\epsilon\in(0,\alpha-\alpha')$.  Also note that, for small enough $\epsilon$, if  \greedy\ has a unique action that maximizes the marginal utility with respect to $\alpha$ (also with respect to not choosing any action), then it also maximizes the marginal utility with respect to $\alpha-\epsilon$ (as long as the set of actions chosen by \greedy\ so far has not changed.). 

    Since by Observation~\ref{obs:must-cand}, every critical $\alpha$ must be a candidate critical $\alpha$ for some set in $\cand{f,c}{a_1,a_2}$, $a_1, a_2 \in \actions \cup \{ \nill \}$, $a_1 \neq a_2$, we can distinguish between the following two cases:

    \textbf{Case 1:} $\alpha \in \cand{f,c}{a, \nill}$ for some action $a\in \actions$ (and thus not in any of $\cand{f,c}{a_1,a_2}$ for $a_1, a_2 \in \actions$ nor in $\cand{f,c}{a', \nill}$ for $a'\neq a$ by genericity of $c$).
		
    We first observe that $a \in \setal$. To see this, assume that $a \not\in \setal$. Then by genericity, $\alpha f(\setal[t] \mid \setal^{t-1}) - c(\setal[t]) > \alpha f(a' \mid \setal^{t-1}) - c(a')$ for all $t$ and all $a' \neq \setal[t]$. As these are finitely many strict inequalities, there is an $\epsilon > 0$ such that we also have $(\alpha - \epsilon) f(\setal[t] \mid \setal^{t-1}) - c(\setal[t]) > (\alpha - \epsilon) f(a' \mid \setal^{t-1}) - c(a')$ for all $t$ and all $a' \neq \setal[t]$. That is, the execution of \greedy\ on $\alpha - \epsilon$ is the same as on $\alpha$: This is a contradiction to $\alpha$ being critical.

    Since $a \in \setal$, let $\ell$ be the step in \greedy\ such that $\setal[\ell]=a$.
		
    One can further assume that the marginal utility is zero, namely that $\alpha \cdot f(a \mid \setal^{\ell-1} ) - c(a) = 0$. Indeed, it must be non-negative (since \greedy\ picked it), and if it is strictly positive then the marginal utilities of all actions up to this point are strictly positive, and there are no ties by the genericity assumption.
    This implies that for small enough $\epsilon$ \greedy\ on $\alpha-\epsilon$ will choose exactly the same set $\setal$, which contradicts $\alpha$ being critical.

    Furthermore, $\alpha \cdot f(a' \mid \setal^{\ell-1} ) - c(a') < 0$ for all $a' \neq a$ because otherwise $\alpha \in\cand{f,c}{a,a'}$, contradicting genericity. Similarly, for $\alpha$, \greedy\ has a positive marginal utility for every step $t<\ell$, and in every such step, the action that maximizes the marginal utility is unique.
    Thus, for small enough $\epsilon$, \greedy\ on $\alpha-\epsilon$ 
    selects the same action $\setal[t]$ for all $t<\ell$, and selects no action at step $\ell$ since the marginal utility of all remaining actions at $\alpha-\epsilon$ is negative.
    It follows that, for $\alpha'$, \greedy\ would have chosen the exact same actions except for $a$, and $\setalp=\setal \setminus\{a\}$.
		
    \textbf{Case 2:} $\alpha \in \cand{f,c}{a_1,a_2}$ for some actions $a_1,a_2\in \actions$. Let $c(a_1) > c(a_2)$. 

    We first show that at least one of $a_1$ and $a_2$ is contained in $\setal$. 
    With the goal of a contradiction, let us assume that $a_1,a_2\notin \setal$. Let $d=|\setal|$. By genericity, for every $t\leq d$, it holds that the selected action $\setal[t]$ in iteration $t$ is the unique optimal action (also with respect to not choosing any action). Therefore, for small enough $\epsilon$, we have $\set_{\alpha - \epsilon} = \setal$, which contradicts $\alpha$ being critical.

    Next we show that indeed $a_1 \in \setal$. To this end, let $\ell$ be the first iteration where one of $a_1$ and $a_2$ is added to $\setal$. That is, $\setal[\ell] \in \{ a_1, a_2 \}$ and $a_1, a_2 \not\in \setal^{\ell-1}$. We observe that $\alpha \cdot f(a_1 \mid \setal^{\ell-1}) - c(a_1) = \alpha\cdot f(a_2 \mid \setal^{\ell-1}) - c(a_2)$ because otherwise due to genericity there is no tie-breaking occurring in \greedy, implying that for small enough $\epsilon$, we have $\setal = \set_{\alpha - \epsilon}$, which contradicts $\alpha$ being critical.  Since \greedy\ breaks ties in favor of actions with higher cost, in iteration $\ell$, \greedy\ chooses to add $a_1$, meaning that $\setal[\ell] = a_1$, thus $a_1 \in \setal$.

    Note also that, since \greedy\ picked $a_1$, the marginal utility $\alpha \cdot f(a_1 \mid S^{\ell-1}_\alpha) -c(a_1) = \alpha \cdot f(a_2 \mid S^{\ell-1}_\alpha) -c(a_2)$ must be non-negative, and it cannot be $0$ either, by genericity and the fact that $a_1,a_2\neq \nill$. Therefore, it must be strictly positive.

    We next consider running \greedy\ on instances in which we perturb the costs of $a_1$ and $a_2$, and relate the obtained outcome to $S_\alpha$ and $S_{\alpha-\epsilon}$. To this end, let
    $Z = \{\alpha \cdot\left(f(a\mid S_1)-f(a'\mid S_2)\right) -c(a)+c(a') \mid a, a' \in \actions, S_1, S_2 \subseteq \actions\} \cup \{\alpha \cdot f(a\mid S) -c(a) \mid a \in \actions, S \subseteq \actions\}$. 
    Furthermore, let 
    $$\delta=\min\left\{\frac{z}{2} \mid z \in Z, z>0\right\}.$$

    Define $c'$ by $c'(a_1) = c(a_1) + \delta$, and $c'(a) = c(a)$ for $a \neq a_1$. Let $X$ be the outcome of \greedy\ on $(f, c', \alpha)$. Because $f$ is gross substitutes, there is a set in the demand $\demandset_{f,c'}(\alpha)$ at $(f,c',\alpha)$ that contains $S_\alpha\setminus \{a_1\}$. Since all steps of \greedy\ on $(f,c',\alpha)$ are unique up to step $\ell$ and in step $\ell$ it adds $a_2$, action $a_2$ must be in every demand set. Combining these two observations implies that $(S_\alpha\setminus \{a_1\}) \cup \{a_2\}$ is contained in a demand set. By genericity and the choice of $\delta$, all steps of \greedy\ on $(f,c',\alpha)$ after step $\ell$ must also be unique, showing that $X \supseteq (S_\alpha\setminus \{a_1\}) \cup \{a_2\}$.

    Note that since all steps of \greedy\ on $(f,c',\alpha)$ are unique, \greedy\ on $(f,c,\alpha-\epsilon)$ gives the same outcome for a sufficiently small $\epsilon$. Thus, $S_{\alpha-\epsilon} = X$.

    Now define $c''$ by $c''(a_2) = c'(a_2) + \delta$, and $c''(a) = c'(a)$ for $a \neq a_2$. Let $Y$ be the outcome of \greedy\ on $(f, c'', \alpha)$. 
    We claim that \greedy\ on $(f,c'',\alpha)$ makes the exact same choices as \greedy\ on $(f,c,\alpha)$. Up to step $\ell-1$, this holds by genericity. At step $\ell$, both $a_1$ and $a_2$ used to be the strict best choices (with a strictly positive marginal) and $a_1$ was chosen by tie-breaking (due to the higher cost). By our choice of $\delta$ and since we raised the cost of both $a_1$ and $a_2$ by the same amount, this is still the case. After step $\ell$, all choices that were unique are still unique. Note that this includes decisions involving $a_2$ because of our choice of $\delta$, and because the only element it could tie with is $a_1$ (by genericity),  and $a_1$ was already chosen.

    We conclude that $Y = S_\alpha$.

    By gross substitutes,  $Y$ must contain $X \setminus \{a_2\}$. 

    Combining everything we have shown so far it must hold that $S_{\alpha}$ contains action $a_1$, $S_{\alpha}$ may or may not contain action $a_2$, and $S_{\alpha}$ may contain some additional actions $x_1, \ldots, x_k$ different from $a_1$ and $a_2$; while $S_{\alpha-\epsilon}$ must contain action $a_2$, $S_{\alpha-\epsilon}$ may or may not contain action $a_1$, $S_{\alpha-\epsilon}$ must contain all other actions $x_1, \ldots, x_k$ contained in $S_\alpha$ (if any), and it may not contain any other action.

    First suppose that $a_2 \in S_{\alpha}$. Note that then we cannot have $a_1 \in S_{\alpha-\epsilon}$ because this would mean that $S_{\alpha-\epsilon} = S_{\alpha}$ in contradiction to $\alpha$ being critical. So we must have $S_{\alpha'} = S_{\alpha-\epsilon} = S_{\alpha}\setminus \{a_1\}$.

    Now consider the case where $a_2 \not\in S_{\alpha}$. In this case we can't have $a_1 \in S_{\alpha-\epsilon}$ because this would mean that $S_{\alpha-\epsilon} \supset S_{\alpha}$ in contradiction to Proposition \ref{prop:monotone}. Hence $a_1 \not\in S_{\alpha-\epsilon}$ and thus $S_{\alpha'} = S_{\alpha-\epsilon} = (S_{\alpha}\setminus\{a_1\}) \cup \{a_2\}$.
\end{proof}

\subsubsection{Arbitrary Additive Cost Functions}
\label{sec:arbitrary-cost}

Finally, we 
extend the bounded critical set result from generic cost functions to arbitrary additive ones.
To do so, we define a perturbation over cost functions that leads to a generic cost function with probability 1 and where the size of the critical set can only increase. 

Given an additive cost function $c$, an additive cost function $\hat{c}$  is said to be an {\em \epsper} of $c$ if for every action $a\in \actions$, $\hat{c}(a) \in [c(a),c(a)+\epsilon]$.

We first observe that a small enough perturbation in the cost function cannot insert new sets into the demand:
\begin{observation}
\label{obs:subset-demand}
	For every success probability function $f$, additive cost function $c$ and $\alpha>0$, there exists $\epsilon>0$ such that for every $\epsilon$-perturbation cost function $\hat{c}$ of $c$, $\demandset_{f,\hat{c}}(\alpha)\subseteq \demandset_{f,c}(\alpha)$.
\end{observation}
\begin{proof}
    Let $\delta = \left( \max_{\set \in \demandset_{f,c}(\alpha)} \alpha \cdot f(\set) -c(\set) \right) - \left( \max_{\set' \notin \demandset_{f,c}(\alpha)} \alpha\cdot  f(\set') -c(\set') \right)$ be the smallest utility gap between a set inside $\demandset_{f,c}$ and outside $\demandset_{f,c}$. Note that $\delta > 0$ since there are only finitely many sets. We now claim that for $\epsilon<\frac{\delta}{|\actions|}$, we have $\demandset_{f,\hat{c}}(\alpha)\subseteq \demandset_{f,c}(\alpha)$. To this end, consider any $\set' \not\in \demandset_{f,c}(\alpha)$. We claim that $\set' \not\in \demandset_{f,\hat{c}}(\alpha)$. Let $\set \in \demandset_{f,c}(\alpha)$. We have
    \begin{align*}
        \alpha \cdot f(\set') -\hat{c}(\set')  &\stackrel{(\star)}{\leq} \alpha \cdot f(\set') -c(\set') \stackrel{(\star\star)}{\leq} \alpha \cdot f(\set) -c(\set) -\delta \\ &\stackrel{(\star)}{\leq} \alpha \cdot f(\set) -\hat{c}(\set)+\epsilon \cdot |\actions| -\delta < \alpha \cdot f(\set) -\hat{c}(\set),
    \end{align*}
    where $(\star)$ follows from the definition of $\hat{c}$ and $(\star\star)$ from the fact that $\set \in \demandset_{f,c}(\alpha)$ but $\set' \not\in \demandset_{f,c}(\alpha)$. 
\end{proof}

With Observation~\ref{obs:subset-demand} at hand, the following lemma shows that a perturbed cost function can only increase the size of the critical set.

\begin{lemma}\label{lem:greater-oracle}
    For every success probability function $f$ and additive cost function $c$, there exists an $\epsilon>0$ such that for every \epsper~$\hat{c}$ of $c$, $|\orbit{f,c}| \leq |\orbit{f,\hat{c}}|$.
\end{lemma}
\begin{proof}
    Let $\alpha_1<\ldots <\alpha_k$ be the critical values  in $\orbit{f,c}$.
    Let $\beta_0=\alpha_1/2$, $\beta_i=\frac{\alpha_i+\alpha_{i+1}}{2}$ for $1 \leq i<k$, and let $\beta_k=2\alpha_k$.
    By Proposition~\ref{prop:monotone} and Corollary~\ref{cor:structure} we get that for every $i$,  and $\set \in \demandset_{f,c}(\beta_i) $, $\g(\alpha_i) \leq f(\set) <\g(\alpha_{i+1})$, and therefore $\{\demandset_{f,c}(\beta_i)\}_{i}$ are  disjoint (i.e., for every $i\neq j$,  $\demandset_{f,c}(\beta_i) \cap \demandset_{f,c}(\beta_j) =\emptyset$).
		
    By Observation~\ref{obs:subset-demand}, there exist $\epsilon_0,\ldots,\epsilon_k$ such that for every $i$,  and $\hat{c_i}$ which is an \epsiper~of $c$, it holds that  
    $\demandset_{f,\hat{c_i}}(\beta_i) \subseteq \demandset_{f,c}(\beta_i)$.
    Thus, for $\epsilon = \min_i \epsilon_i$ it holds that for every \epsper~$\hat{c}$ of $c$, $\demandset_{f,\hat{c}}(\beta_i) \subseteq  \demandset_{f,c}(\beta_i)$, and thus $\{\demandset_{f,\hat{c}}(\beta_i)\}_{i}$ are disjoint. 
    Therefore, every interval $(\beta_{i-1},\beta_{i})$ must have a critical $\hat{\alpha}_i$ w.r.t. $\hat{c}$. 
    This concludes the proof.
\end{proof}

We are now ready to prove Theorem~\ref{theorem:gs}, namely to establish the upper bound of $\frac{n(n+1)}{2}$ on the size of the critical set for an arbitrary additive cost function.

\begin{proof}[Proof of Theorem~\ref{theorem:gs}]
    By Lemma~\ref{lem:greater-oracle} there exists $\epsilon>0$ such that for every \epsper~$\hat{c}$ of $c$, it holds that $|\orbit{f,c}| \leq |\orbit{f,\hat{c}}|$.
    Suppose one draws $\hat{c}(a)$ uniformly at random from $[c(a),c(a)+\epsilon]$ for every action $a$. 
    We show that $\hat{c}$ would be generic with probability 1. 
	
    To see this, consider the event that $\cand{f,c}{a_1,a_2}\cap\cand{f,c}{a_3,a_4}$ contains some $\alpha>0$ for two different (unordered) pairs $(a_1, a_2),(a_3,a_4)$. 
    By the union bound, it holds that 
    \begin{eqnarray}
	& & \Pr\Bigl[\cand{f,c}{a_1,a_2}\cap\cand{f,c}{a_3,a_4}\neq \emptyset\Bigr] 
        \nonumber \\ & \leq & \sum_{\set_1,\set_2,\set_3,\set_4} \Pr\Biggl[\exists \alpha>0 : \alpha\cdot f(a_1 \mid \set_1) -\hat{c}(a_1)=\alpha \cdot f(a_2 \mid \set_2) -\hat{c}(a_2) \label{eq:a1a2} \\  &  & \phantom{xxxxxxxxxxx}  \textsc{and } \alpha\cdot f(a_3 \mid \set_3) -\hat{c}(a_3)=\alpha \cdot f(a_4 \mid \set_4) -\hat{c}(a_4)  \Biggr] \nonumber.
    \end{eqnarray}

    If $f(a_1 \mid \set_1)=f(a_2 \mid \set_2)$, then 
    $$
        \Pr\Bigl[\cand{f,c}{a_1,a_2}\cap\cand{f,c}{a_3,a_4}\neq \emptyset\Bigr] \stackrel{\eqref{eq:a1a2}}{\leq}  \Pr[\hat{c}(a_1)=\hat{c}(a_2)] = 0, 
    $$
    where the last equality follows since this is a measure 0 event.
    Else,
    \begin{align*}
        & \Pr\Bigl[\cand{f,c}{a_1,a_2} \cap\cand{f,c}{a_3,a_4}\neq \emptyset\Bigr] \\  &\phantom{xxxxxx} \stackrel{\eqref{eq:a1a2}}{\leq}   \sum_{\substack{\phantom{,}\set_1,\set_2,\\\set_3,\set_4}} \Pr\Biggl[ \left[f(a_1 \mid \set_1)-f(a_2 \mid \set_2)\right]\left[\hat{c}(a_3)-\hat{c}(a_4)\right] = \\[-8pt] & \phantom{xxxxxxxxxxxxxxxxxxxxx}\left[f(a_3 \mid \set_3)-f(a_4 \mid \set_4)\right]\left[\hat{c}(a_1)-\hat{c}(a_2)\right]\Biggr] = 0,
    \end{align*}
    where the last equality follows again since this is a measure 0 event. 
    Applying the union bound once more on all such (finitely many) events shows that $\hat{c}$ is generic with probability 1.
    Thus, there exists a generic $\epsilon$-perturbation $\hat{c}$ of $c$.
    We get that
    $|\orbit{f,c}| \leq |\orbit{f,\hat{c}}| \leq \frac{n (n+1)}{2}$, where the last inequality follows by Lemma~\ref{lem:oracle-generic}.
    This concludes the proof.
\end{proof}

\section{Beyond Gross Substitutes Functions}
\label{sec:beyond}

In this section we study success probability functions beyond gross substitutes. 
In Section~\ref{sec:submodular} we show that submodular functions are more complex than gross substitutes: unlike gross substitutes, they may exhibit an exponential critical set, and the optimal contract is NP-hard to compute. 
In Section~\ref{sec:ptas} we devise an FPTAS for arbitrary success probability functions.
In Section~\ref{sec:weakly-poly} we present a weakly poly-time algorithm for instances with poly-size critical sets.

\subsection{Submodular Functions}
\label{sec:submodular}

Our first result shows that our approach for gross substitutes functions of iterating over all critical points cannot yield a polytime algorithm for submodular functions or even for coverage functions.

\begin{theorem}
    There exists a coverage success probability function $f$ and an additive cost function $c$ such that 
    $|\orbit{f,c}|=2^{n}-1$.
\end{theorem}

\begin{proof}
    For simplicity of presentation, we present a proof for $f$ values that are not necessarily in $[0,1]$; this can be easily scaled. 

    We prove the theorem by induction on the size of the action set. 
    For $n = 1$, it is trivial. (E.g., $f(1)=2$, and $c(1)=1$, gives $|\orbit{f,c}|=1$.)

    Assume there exist a coverage success probability function $f$ and an additive cost function $c$ over a set of actions $\actions$ of size $k$ such that $|\orbit{f,c}|=2^k-1$, and let $\alpha_1 < \ldots < \alpha_{2^k-1}$ be the critical values of $\alpha$ in $\orbit{f,c}$.

    Let $g$ be a success probability function over the set of actions $\actions \cup\{k+1\}$, given by
    $$
        g(S)=
        \begin{cases}
            \beta_1 \cdot f(S) & \mbox{if } S \subseteq \actions\\
            \beta_2 \cdot f(\actions) +f(S\setminus\{k+1\}) & \mbox{if } k+1 \in S
        \end{cases}
    $$
    where $\beta_1 = \frac{10 \cdot \alpha_{2^k-1}}{\alpha_1}$ and $\beta_2=10 \cdot \beta_1$.
    The function $g$ is a coverage function (see Claim~\ref{cl:submodular} in the appendix).
    Let $\hat{c}$ be the additive cost function over actions in $\actions \cup\{k+1\}$ defined as $\hat{c}(a) =c(a)$ for every $a\in \actions$, and $\hat{c}({k+1}) = 20\cdot \alpha_{2^k-1}\cdot f(\actions)$.
    
    We show that there are $2^k-1$ critical values of $\alpha$ in the range $\alpha \leq \frac{\alpha_{2^k-1}}{\beta_1}$, a single critical value of $\alpha$ in the range $(\frac{\alpha_{2^k-1}}{\beta_1},\frac{\alpha_1}{2})$, and $2^k-1$ additional critical values of  $\alpha$ in the range $\alpha > \frac{\alpha_1}{2}$, amounting to $2^{k+1}-1$ critical $\alpha$'s.
    
    For every $\alpha \leq \frac{\alpha_{2^k-1}}{\beta_1}$ the marginal utility of action $k+1$ with respect to any set $S \subseteq \actions$ is at most 
    \begin{eqnarray*}
        & & \alpha\left(\beta_2 f(\actions) + f(S) - \beta_1 f(S)\right) - \hat{c}(k+1) \\
        & \leq & 
        \alpha \cdot \beta_2 \cdot f(\actions) -\hat{c}(k+1) \\
        & \leq & 10\cdot \alpha_{2^k-1} \cdot f(\actions) - 20\cdot \alpha_{2^k-1}\cdot f(\actions) < 0, 
    \end{eqnarray*}
    where the first inequality follows by $\beta_1 \geq 1$, and the second inequality follows by the range of $\alpha$ and by substituting $\beta_2/\beta_1=10$.
    It follows that for every $\alpha \leq \frac{\alpha_{2^k-1}}{\beta_1}$, action $k+1$ is never included in any demanded set, thus $g(S)=\beta_1 f(S)$. Therefore,  for every $i=1,\ldots,2^k-1$, $\frac{\alpha_i}{\beta_1}$ is a critical $\alpha$ .
    
    We next show that in the range $(\frac{\alpha_{2^k-1}}{\beta_1},\frac{\alpha_1}{2})$, there must exists an additional critical $\alpha$.
    By Corollary~\ref{cor:structure}, every critical $\alpha$ leads to a demanded set of a strictly higher value. Thus, the demanded set at $\alpha=\frac{\alpha_{2^k-1}}{\beta_1}$ must be $\actions$.
    We now show that for $\alpha= \alpha_1/2$, action $k+1$ must be in every set in the demand. 
    Indeed, the utility from action $k+1$ alone is 	
    \begin{eqnarray*}
        & & \alpha\cdot \beta_2 \cdot f(\actions) -\hat{c}(k+1) \\  
        & = &  \frac{\alpha_1}{2} \cdot \frac{100 \cdot \alpha_{2^k-1}}{\alpha_1} \cdot f(\actions) -20\cdot \alpha_{2^k-1}\cdot f(\actions) \\
        & = & 30 \cdot \alpha_{2^k-1} \cdot f(\actions) ,
    \end{eqnarray*}
    while the utility from any set that does not contain action $k+1$ is at most  
    $$\alpha \cdot \beta_1\cdot  f(\actions) = \frac{\alpha_1}{2} \cdot \frac{10 \cdot \alpha_{2^k-1}}{\alpha_1} \cdot f(\actions) = 5 \cdot   \alpha_{2^k-1} \cdot f(\actions)  .$$
    The same argument shows that $k+1$ must be in the demand of every set for $\alpha>\frac{\alpha_1}{2}$.
    We conclude that there must exists a critical alpha in the range $(\frac{\alpha_{2^k-1}}{\beta_1},\frac{\alpha_1}{2})$ and that action $k+1$ must be in the demand of this critical.
    Moreover, for $\alpha=\frac{\alpha_1}{2}$, the demand is exactly $\{k+1\}$.
    This is since for every action $a\neq k+1$ it holds that  $\alpha f(a \mid \{k+1\}) -c(a) = \alpha f(a)-c(a)  <0$, since $\alpha<\alpha_1$
    Thus, all other actions has a negative marginal utility since, and are not in the demand set.
    
    For $\alpha > \frac{\alpha_1}{2}$, for every set $S$, the marginal utility of $S$ with respect to the action $k+1$ is the same for $f$ and $g$. Thus, every $\alpha_i$ in this range is also critical. 
\end{proof}

Our next theorem establishes NP-hardness for computing an optimal contract under submodular functions, using a reduction from subset sum. Concretely, it is NP-hard to compute the optimal contract for the case that function $f$ is budget additive, meaning that it is given by integer values $f(\{1\}), \ldots, f(\{n\})$ and $B$ such that $f(S) = \min\{B, \sum_{i \in S} f(\{i\})\}$.

\begin{theorem} \label{thm:nphard}
    The optimal contract problem for budget-additive (and, hence, submodular) success probability functions and additive cost functions is NP-hard.
\end{theorem}

\begin{proof}
    We prove the theorem by a reduction from \textsc{Subset-Sum}.
    This problem receives as input a (multi-)set of positive integer values $X=\{x_1,\ldots,x_n\}$ and an integer value $Z$. The question is whether there exists a subset $S \subseteq X$ such that $\sum_{j\in S} x_j = Z$.
    W.l.o.g., assume that $x_i<Z$ for all $i$ (all numbers greater than $Z$ can be ignored), and that $\sum_{i\in X} x_i>Z$ (otherwise this is an easy instance).
	
    Given an instance $(x_1,\ldots,x_n,Z)$ to subset-sum, construct an instance to the optimal contract problem for budget additive functions over $n$ actions as follows.  
    For every action $i=1,\ldots,n$, set $f(\{i\})=x_i$, and set $B=Z$. I.e., for every set $S$, $f(S)=\min(Z,\sum_{i\in S} x_i)$. 
    Let the cost function be $c(i)=\epsilon \cdot x_i$, where $\epsilon=\frac{1}{Z^2}$.
	
    If there exists a set $S$ such that $\sum_{i \in S} x_i = Z$, then for a contract of $\alpha\geq\epsilon$ the agent's best response is the set $S$, and for $\alpha<\epsilon$ the agent's best response is the empty-set.
    Thus, the optimal contract is to set $\alpha=\epsilon$ where the principal utility is $(1-\epsilon)\cdot Z$.

    Consider next the case where there does not exist a set $S$ such that $\sum_{i \in S} x_i = Z$. 
    Let $Z_1 =\argmin \{z > Z \mid  \exists S\subseteq [n]. \sum_{i \in S} x_i =z \}$, and let $S_1$ be the set that sums to $Z_1$.
    Similarly, let $ Z_2 =\arg\max \{z < Z \mid \exists S\subseteq [n]. \sum_{i \in S} x_i =z \}$, and let $S_2$ be the set that sums to $Z_2$.

    Every set $S$ such that $\sum_{i \in S} x_i > Z$ gives an agent's utility of $\alpha Z -\epsilon \sum_{i \in S} x_i$. Thus, $S_1$ is optimal among all these sets. 
    Similarly, every set $S$ such that $\sum_{i \in S} x_i < Z$ gives an agent's utility of $(\alpha-\epsilon) \sum_{i \in S} x_i$. Thus, for $\alpha\geq \epsilon$, $S_2$ is optimal among all these sets.
    It follows that there are exactly two critical $\alpha$'s, namely $\alpha_1=\epsilon$, where the agent selects $S_2$ and the principal's utility is $(1-\epsilon)Z_2$, and $\alpha_2=\frac{Z_1-Z_2 }{Z-Z_2}\cdot \epsilon$, where the agent selects $S_1$ and the principal's utility is $(1-\frac{Z_1-Z_2 }{Z -Z_2}\cdot \epsilon) Z$.

    We claim that the latter contract is better for the principal. 
    To see this, observe first that $Z_1-Z_2<Z$. Indeed, let $i$ be an arbitrary action in $S_1 \setminus S_2$ (such an action must exist). It holds that $Z_2+x_i > Z$ (else, contradicting the choice of $Z_2$). It further follows that $Z_2+x_i \geq Z_1$ (else, contradicting the choice of $Z_1$). On the other hand, $x_i<Z$. We get $Z_1 \leq Z_2 + x_i < Z_2 + Z$, as claimed.
    Furthermore, $Z-Z_2 \geq 1$. It follows that $\frac{Z_1-Z_2 }{Z -Z_2}<Z$. We get
    $
    (1-\frac{Z_1-Z_2 }{Z -Z_2} \epsilon) Z > (1-Z \epsilon) Z > (1-\epsilon) Z_2, $
    where the last inequality follows by $Z_2 \leq Z-1$ and $\epsilon = 1/Z^2$. 
    Thus, the principal's utility in this case is $(1-\frac{Z_1-Z_2 }{Z -Z_2}\cdot \epsilon) Z$, which is strictly smaller than $(1-\epsilon)Z$. 

    It follows that the optimal contract is $\alpha=\epsilon$ if and only if the \textsc{Subset-Sum} instance is a YES instance.
\end{proof}

\subsection{FPTAS for Arbitrary Functions}
\label{sec:ptas}

We next present an FPTAS for the optimal contract problem with general success probability $f$ and additive costs $c$, with value and demand oracle access. 
Our result is best possible in two ways: First, we know that computing an optimal contract is \textsf{NP}-hard for submodular $f$ and additive $c$ (see Theorem~\ref{thm:nphard}).
Second, it was recently shown that computing an optimal contract for submodular $f$ and additive $c$ requires exponentially many demand oracle calls \cite{DuettingFGT24}.

\begin{theorem}[FPTAS with demand oracle]\label{thm:FPTAS}
    For a general success probability function $f$ and additive cost function $c$,
    Algorithm~\ref{alg:fptas} gives a $(1-\epsilon)$-approximation to the optimal principal utility with $O\left(\frac{n^2}{\epsilon} \right)$ many value and demand queries to the success probability function $f$.
\end{theorem}

The main challenge in proving Theorem~\ref{thm:FPTAS} is that the optimal $\alpha^\star$ could be very close to $1$. Indeed, if we knew that $\alpha^\star$ was bounded away from $1$, we could find a close to optimal contract through a fine-enough discretization of the contract space.
To bound $\alpha^\star$, we use that, by a reduction to the non-combinatorial model of \cite{DBLP:conf/ec/DuttingRT19}, the gap between the maximum welfare and the principal's utility under the optimal contract is at most $2^n$. 

\begin{observation}[D\"utting et al.~\cite{DBLP:conf/ec/DuttingRT19}]\label{obs:welfare-single}
    For a general success probability function $f$ and additive
    cost function $c$, there exists a contract $\alpha$ that guarantees a utility for the principal of at least  $\frac{OPT}{2^n}$, where $OPT=\max_{S} (f(S)-c(S))$.
\end{observation}

Using this observation, we now show our key lemma, which sets lower and upper bounds on $\alpha^\star$.
\begin{lemma}\label{lem:4m}
    For a general success probability function $f$ and an additive cost function $c$,
    let $\alpha^\star $, $S^\star $ be the optimal contract and the optimal set of actions 
    respectively, let $j^\star \in \arg\max_{j \in S^\star} c(j)$, and let $OPT=\max_{S} (f(S)-c(S)) > 0$.
    Then we have
    $\alpha_{\min} \leq \alpha^\star \leq  \alpha_{\max}$, where $\alpha_{\min} =1- \frac{OPT}{ c({j^\star}) + OPT}$ and $\alpha_{\max} = 1- \frac{OPT}{n\cdot 2^n (c({j^\star}) + OPT )}$.
\end{lemma}

\begin{proof}
    The set that maximizes $f(S)-c(S)$ is also the best response of the agent under contract $\alpha=1$.
    We note that since $OPT>0$, it must be that $\alpha^\star <1 $.
    By Observation~\ref{obs:welfare-single}, it holds that 
    \begin{equation}
        \frac{OPT}{2^n}\leq (1-\alpha^\star)  f(S^\star)   \leq   f(S^\star) -c(S^\star) \leq OPT,\label{eq:welfare3}
    \end{equation}
    where the second inequality is since the utility of the principal under contract $\alpha^\star$ is bounded by the welfare generated by $S^\star$ as the utility of the agent is non-negative.
    
    We first show that $\alpha^\star \geq \alpha_{\min}$. It holds that 
    $$ 
        \alpha^\star \cdot (  c(S^\star) + OPT  )-c(S^\star) \geq \alpha^\star \cdot (  c(S^\star) + f(S^\star)-c(S^\star) )-c(S^\star) =  \alpha^\star \cdot f(S^\star) -c(S^\star) \geq 0, 
    $$
    where the first inequality is by Inequality~\eqref{eq:welfare3}, and the last inequality is since the utility of the agent under $\alpha^\star$ and $S^\star$ is non-negative.

    By rearranging, we get that 
    $$ 
        \alpha^\star \geq \frac{c(S^\star)}{ c(S^\star) + OPT} \geq \frac{c({j^\star})}{c({j^\star}) +OPT } = 1- \frac{OPT}{ c({j^\star}) + OPT},
    $$
    as needed.
    
    To show that $\alpha^\star \leq \alpha_{\max}$, observe that
    $$ 
        (1-\alpha^\star) (c(S^\star) + OPT) \stackrel{\eqref{eq:welfare3}}{\geq }  (1-\alpha^\star) (c(S^\star) + f(S^\star) - c(S^\star)) =  (1-\alpha^\star) f(S^\star) \stackrel{\eqref{eq:welfare3}}{\geq } \frac{OPT}{2^n}.
    $$
    Rearranging this, we obtain
    $$  
        \alpha^\star \leq 1- \frac{OPT}{2^n (c(S^\star) + OPT )} \leq 1- \frac{OPT}{n\cdot 2^n (c({j^\star}) + OPT )}  ,
    $$
    which concludes the proof.
\end{proof}

\begin{algorithm}
    \caption{FPTAS for a single agent using value and demand oracles}\label{alg:fptas}
   \hspace*{\algorithmicindent} \textbf{Parameter:}  $\epsilon \in (0,1) $ \\
   \hspace*{\algorithmicindent} \textbf{Input:}  Costs $c(1),\ldots,c(n) \in \reals_{\geq 0}$, value and demand oracle access to a function $f:2^A \rightarrow [0,1]$  \\
    \hspace*{\algorithmicindent} \textbf{Output:}  A contract $\alpha$ 
    \begin{algorithmic}[1]
    \STATE Let $\alpha =0$ and $S = \arg\max_{\hat{S}: c(\hat{S}) = 0} f(\hat{S})$ \label{st:init}
    \STATE Let $OPT = \max_{\hat{S} \subseteq \actions} (f(\hat{S}) - c(\hat{S}))$
    \FOR{$j\in\actions$ with $c(j)>0$}
    \FOR{$k=0,\ldots,\lceil \log_{1/(1-\epsilon)}n \cdot 2^n \rceil $}
    \STATE Let $\alpha_{j, k} = 1 - (1-\epsilon)^{k+1} \cdot  \frac{OPT}{ c({j}) + OPT}$
    \STATE Let $S_{j,k} = \arg\max_{\hat{S}\subseteq \actions}  \left(f(\hat{S}) -\sum_{j'\in \hat{S}} \frac{c({j'})}{\alpha_{j,k}}\right)$ 
    \IF{$(1-\alpha_{j, k}) f(S_{j,k}) > (1-\alpha)f(S)$}
    \STATE $\alpha=\alpha_{j, k}$
    \STATE $S= S_{j,k}$
    \ENDIF
    \ENDFOR
    \ENDFOR
    \RETURN $\alpha$ 
    \end{algorithmic}
\end{algorithm}
 
We are now ready to prove Theorem~\ref{thm:FPTAS}.

\begin{proof}[Proof of Theorem~\ref{thm:FPTAS}]
    Let $\alpha^\star,S^\star$ be the best contract with its best response. Note that if $u_p(S^\star,\alpha^\star) \leq 0$, the claim holds trivially because Algorithm~\ref{alg:fptas} ensures that $u_p(S, \alpha) \geq 0$. Otherwise, $u_p(S^\star,\alpha^\star) = (1-\alpha^\star)f(S^\star) > 0$. It must then hold that $\alpha^\star < 1$ and $f(S^\star) > 0$. 
    So, in particular, $S^\star \neq \emptyset$, so that $j^\star \in \arg\max_{j \in S^\star} c(j)$ is well defined.  Also note that $OPT = \max_{\hat{S} \subseteq A} (f(\hat{S}) - c(\hat{S})) \geq f(S^\star) - c(S^\star) \geq (1-\alpha^\star) f(S^\star) >0$, because the agent's utility from $S^\star$ is non-negative. 
    
    If $c({j^\star})=0$, then $c(S^\star) = 0$ and the optimal contract is $\alpha=0$, which is the contract considered in Step~\ref{st:init} of the algorithm. 
    Otherwise, $c({j^\star}) > 0$. Consider the iteration of Algorithm~\ref{alg:fptas} in which $j = j^\star$.
    Let us denote $\alpha_{\min} =1- \frac{OPT}{ c({j^\star}) + OPT}$ and $\alpha_{\max} = 1- \frac{OPT}{n\cdot 2^n (c({j^\star}) + OPT )}$.
    We claim that then there must be a value of  $k \in \{0,\ldots,\lceil \log_{1/(1-\epsilon)}n \cdot 2^n \rceil\}$ such that $1 - \alpha_{j, k}  \leq 1 - \alpha^\star \leq  \frac{1 - \alpha_{j, k}}{1-\epsilon}$.
    Indeed, for $k = 0$ we have $\frac{1-\alpha_{j^\star,0}}{1-\epsilon} =  1- \alpha_{\min} \geq  1-\alpha^\star$, where the inequality follows by Lemma~\ref{lem:4m}, while for $k =\lceil \log_{1/(1-\epsilon)}n\cdot 2^n \rceil$ we have $1-\alpha^\star \geq 1- \alpha_{\max} \geq 1-\alpha_{j^\star,\lceil \log_{1/(1-\epsilon)}n\cdot 2^n \rceil}$, where the first inequality follows again by Lemma~\ref{lem:4m}. So there must be a $k \in \{0,\ldots,\lceil \log_{1/(1-\epsilon)}n \cdot 2^n \rceil\}$ with the desired properties.
    
    We claim that for this choice of $j,k$, contract $\alpha_{j, k}$ provides a $(1-\epsilon)$-approximation to the optimal contract.
    To see this, let $S$ be the choice of the agent under $\alpha_{j, k}$.
    Using Proposition~\ref{prop:monotone}, since $\alpha_{j, k} \geq \alpha^\star$, it must hold that $f(S) \geq f(S^\star)$. We thus obtain,
    \[
        (1-\alpha_{j, k}) f(S) \geq (1-\alpha_{j, k}) f(S^\star) \geq  (1-\epsilon)(1-\alpha^\star) f(S^\star) ,
    \]
    which completes the proof.
\end{proof}

\begin{remark}
    \emph{We remark that Theorem~\ref{thm:FPTAS} can be adapted to beyond additive cost functions. 
    In particular, under subadditive cost functions, if the principal is assumed to have access to a best-response oracle (i.e., given, a contract $\alpha$, what is the set of actions that maximizes the agent's utility), then Algorithm~\ref{alg:fptas} (when replacing Steps 2 and 6 with corresponding calls to a best-response oracle) finds a $(1-\epsilon)$-approximation in time polynomial in $n$ and $1/\epsilon$.
    The only differences in the proof (when considering subadditive cost functions) are in the last derivation of the proof of Lemma~\ref{lem:4m} and the argument for the $c({j^\star}) = 0$ case in the proof of Theorem~\ref{thm:FPTAS}, which now hold by subadditivity.
}
\end{remark}

\subsection{Weakly Poly-Time Algorithm for Instances with Poly-Size Critical Sets}
\label{sec:weakly-poly}

Finally we show that in cases where the critical set is of polynomial size, one can compute the optimal contract in weakly polynomial time. Recall that the number of $\g_{f,c}$ queries required by Algorithm~\ref{alg:contract} is the size of $|\orbitfc|$ multiplied by the number of $\g_{f,c}$ queries required to implement $\successor_{f,c}$. 

We show that if all $f$ and $c$ values are multiples of $2^{-k}$,  $\successor_{f,c}$ can be implemented using $2k+1$ queries to a $\g_{f,c}$ oracle.
Thus, Algorithm~\ref{alg:contract} can be computed using $(2k+1) |\orbitfc|$ queries. 
 
\begin{theorem}\label{thm:2k1}
    Consider the case where all values of $f$ and $c$ are multiples of $2^{-k}$.
    Then, $\successor_{f,c}(\alpha)$ can be computed using $2k + 1$ queries to a $\g_{f,c}$ oracle.
\end{theorem}

\begin{proof}
    To compute $\successor_{f,c}(\alpha)$, we return $\nill$ if $\g(1) = \g(\alpha)$, otherwise, we return the result of \textsc{Search}($\alpha$, $1$).

\bigskip

\begin{algorithm}[H]
    \caption{\textsc{Search}($\alpha_L$, $\alpha_R$)}
    \label{alg:search}
    \begin{algorithmic}
    \IF {$\alpha_R - \alpha_L \leq \frac{1}{2^{2k}}$}
    \STATE return the unique element of $\ratios \cap (\alpha_L, \alpha_R]$, e.g., by the algorithm of Kwek and Mehlhorn \cite{DBLP:journals/ipl/KwekM03}
    \ELSE \IF{$\g(\frac{\alpha_L + \alpha_R}{2}) > \g(\alpha_L)$} \STATE \textbf{return} \textsc{Search}($\alpha_L$, $\frac{\alpha_L + \alpha_R}{2}$) \ELSE \STATE \textbf{return} \textsc{Search}($\frac{\alpha_L + \alpha_R}{2}$, $\alpha_R$)
    \ENDIF
    \ENDIF
\end{algorithmic}	
\end{algorithm}

\bigskip

    \textsc{Search}($\alpha_L$, $\alpha_R$) returns the smallest $\alpha \in \orbitfc \cap (\alpha_L, \alpha_R]$. Note that $\orbitfc \cap (\alpha_L, \alpha_R]$ is non-empty if and only if $\g(\alpha_R) > \g(\alpha_L)$. This is maintained as an invariant.

    Observe that $\alpha_R - \alpha_L$ halves in every step recursive step. Therefore, after at most $2k$ recursive calls we have $\alpha_R - \alpha_L \leq \frac{1}{2^{2k}}$ and the algorithm terminates.

    It remains to show correctness of the algorithm. First, observe that $\successor_{f,c}(\alpha) \in (\alpha_L, \alpha_R]$ at all times. Furthermore, when $\alpha_R - \alpha_L \leq \frac{1}{2^{2k}}$, there is a unique element in $\ratios \cap (\alpha_L, \alpha_R]$ and it is in $\orbitfc$. This is for the following reasons: (i) Suppose $\frac{a}{b} < \frac{a'}{b'} \in \ratios \cap (\alpha_L, \alpha_R]$, then we have $\frac{a'}{b'} -\frac{a}{b} = \frac{a'b - ab'}{bb'} \geq \frac{1}{2^{2k}}$. This is a contraction to $\alpha_R - \alpha_L \leq \frac{1}{2^{2k}}$. So, there cannot be more than one element inside $\ratios \cap (\alpha_L, \alpha_R]$. (ii) $\orbitfc \cap (\alpha_L, \alpha_R]$ has to be non-empty because otherwise $\g(\alpha_L) = \g(\alpha_R)$.
\end{proof}

Note that this algorithm only runs in weakly polynomial time as its running time depends on $k$. Subsequent to the first publication of this result, algorithms running in strongly polynomial were given by Deo-Campo Vuong et al.~\cite{DeoCampoVuongEtAl24} and D\"utting et al.~\cite{DuettingFGT24}.

\section{Beyond Binary Rewards}
\label{sec:linear-optimal}

In this section we study a generalization of the binary outcome model, where the outcome space is $\Omega = \{0,\ldots,m-1\}$. Each outcome $j \in \Omega$ is associated with a non-negative reward $r(j)$. We index outcomes so that $r(0) \le r(1) \le \ldots \le r(m-1)$. Every set of actions $\set \subseteq \actions$ entails some probability distribution over rewards, with $f_j(S)$ being the probability of outcome $j$ under actions $S$.
We use $R(\set)=\sum_{j} f_j(S) \cdot r(j)$ to denote the expected reward to the principal given action set $\set$. We assume that $R$ is monotone and normalized (i.e., $R(S) \geq R(S')$ for every $S'\subseteq S$, and $R(\emptyset)=0$). Note that this implies that $r(0) = \min_{j \in \Omega} r(j) = 0$.
A contract in this model is a function $t: \Omega \rightarrow \reals_{\geq 0}$, specifying the payment from the principal to the agent for every observed outcome.
A contract is said to be {\em linear} if there exists some $\alpha \in [0,1]$ such that $t(j)=\alpha \cdot r(j)$ for every  $j$.

The agent's utility for a set of actions $\set$ under contract $t$ is $\sum_j f_j(S) \cdot t(j) - c(S)$, while the principal's utility is $\sum_j f_j(S) \cdot r(j) - \sum_j f_j(S) \cdot t(j)$. We assume that the agent chooses a set of actions $S$ that maximizes his utility, breaking ties in favor of the principal. Note that in a linear contract with $\alpha < 1$ this is equivalent to assuming that the agent favors sets with higher expected reward.

As in the binary outcome case, we can make some structural assumptions on the expected principal's reward $R(\set)$. 
For example, that $R(\set)$ is submodular or gross substitutes.  
We show that when restricting attention to linear contracts, all our (positive and negative) results from the binary outcome model continue to hold in the general model. 
Moreover, we show that linear contracts are {\em max-min optimal} among all possible contracts, meaning that if the principal knows the expected reward $R(S)$ for every set of actions $S \subseteq A$, but not the probability distribution over rewards obtained from $S$, then a linear contract maximizes the principal's utility in the worst case over all distributions compatible with the known expected rewards. The latter is a direct corollary of Theorem~1 in \cite{DBLP:conf/ec/DuttingRT19}.

\subsection{Robust Optimality of Linear Contracts}
\label{sec:robust}
We start by showing that linear contracts are robustly optimal in the general model.
To do so we introduce some notation.
Let $\mathcal{D}$ denote the collection of all sets of probability distributions over outcomes $j \in \Omega$ with (fixed) non-negative rewards $r(j)$ that are compatible with the known expected rewards $R(\set) \in [r(0),r(m-1)]$ for $\set \subseteq \actions$.
I.e., $\mathcal{D}$ is the collection of 
all sets of probability distributions
$\{D_\set\}_{\set\subseteq \actions}$ over outcomes such that $\E_{j \sim D_\set}[r(j)]=R(\set)$ for every $\set\subseteq \actions$.

Let $\utilmin{t}$ denote the worst-case  principal's utility under contract $t$, over all sets of probability distributions in $\mathcal{D}$, fixing expected rewards $\{R(S)\}_{S \subseteq A}$.
I.e., 
\[
    \utilmin{t} = \min_{D\in \mathcal{D}} \E_{ j \sim D_{S_{t}}} [r(j) - t(j)],
\] 
where $S_t$ is the agent's best response action set for contract $t$.
The following theorem shows that there exists a linear contract that maximizes $\utilmin{t}$ among all contracts $t$.

\begin{theorem}[cf.~Theorem 1 in \cite{DBLP:conf/ec/DuttingRT19}]\label{th:robust-optimality}
    For every function $R$ of expected rewards over rewards $0 = r(0) \le r(1) \le \ldots \le r(m-1)$, additive cost function $c$, and contract $t$, there exists a linear contract $\alpha$ such that $\utilmin{t} \leq \utilmin{\alpha}$.
\end{theorem}

\begin{proof}
    We first note that the best response of the agent to a linear contract $\alpha$ only depends on the reward function $R$, and not on the specific compatible collection of distributions. To see this, observe that the agent's utility for taking the set of actions $S$ is $\alpha\cdot R(S) -c(S)$. Note that this implies that the principal's utility under the agent's best response does not depend on the compatible collection of distributions either. Indeed, under the agent's best response $S_{\alpha} \in \arg\max (\alpha \cdot R(S) -c(S))$, 
    the principal's utility is $(1-\alpha) \cdot R(S_{\alpha})$. Thus, in order to prove the theorem, it suffices to show that for every $R$ and every contract $t$ there exists a compatible collection of distributions and a linear contract $\alpha$ such that the principal's utility under $\alpha$ is at least as high as the utility under contract $t$.
    
    To this end, fix $R$ and $t$, and let $\ell_0=t(0)$ and $\ell_1=t(m-1)$.
    Consider the collection of distributions $\{D_S\}_{S \subseteq \actions}$ over the two outcomes $\{0,m-1\}$, where $\Pr_{j \sim D_S}[j = m-1] = \frac{R(S)}{r(m-1)}$, and $\Pr_{j \sim D_S}[j = 0] = 1-\frac{R(S)}{r(m-1)}$. Observe that this 
    set of probability distributions is compatible with $R$.
    
    Let $S^\star$ be the best response action set of the agent to contract $t$ under distributions $\{D_S\}_{S \subseteq \actions}$. For every $\set$, it holds that
    \begin{align}
        \frac{R(S^\star)}{r(m-1)} \cdot \ell_1 + (1-&\frac{R(S^\star)}{r(m-1)}) \cdot \ell_0 - c(S^\star) \geq \label{eq:ua1}\\
        &\frac{R(S)}{r(m-1)} \cdot \ell_1 +(1-\frac{R(S)}{r(m-1)}) \cdot\ell_0 - c(S)  \notag.  
    \end{align}

    If $\ell_0 > \ell_1$, then we must have $R(S^\star) = 0$ (otherwise the LHS of Equation~(\ref{eq:ua1}) is strictly smaller than $\ell_0 - c(S^\star) \leq \ell_0$, while the RHS is at least $\ell_0$ by choosing $S = \emptyset$). 
    But since the payments are always non-negative, the principal's utility under this case is bounded by $R(S^\star)\leq 0$, while any linear contract $\alpha \in [0,1]$ yields a non-negative utility for the principal so $\utilmin{t}\leq \utilmin{\alpha}$, which concludes the proof of this case.
	
    Else ($\ell_1\geq \ell_0$), consider the linear contract  
     $\alpha= \frac{\ell_1-\ell_0}{r(m-1)}$.
    Observe that $S^\star$ is a best response to $\alpha$, since for every $S$ it holds that:
    \begin{align*}
        \frac{R(S^\star)}{r(m-1)} \cdot \alpha \cdot r(m -1) & - c(S^\star) =  \frac{R(S^\star)(\ell_1-\ell_0)}{r(m-1)} -c(S^\star) \\
        &\stackrel{\eqref{eq:ua1}}{\geq}\frac{R(S)(\ell_1-\ell_0)}{r(m-1)} -c(S) =  \frac{R(S)}{r(m-1)} \cdot \alpha \cdot r(m-1) -c(S),
    \end{align*}
    where the inequality follows by rearranging Equation~\eqref{eq:ua1}.
    Therefore, the expected utility of the principal under the linear contract  $\alpha$ is at least
    \begin{align*}
        \frac{R(S^\star)}{r(m-1)} \cdot r(m-1)\cdot (1-\alpha)  &= R(S^\star) \cdot(1-\frac{\ell_1-\ell_0}{r(m-1)}) \\ &\geq \frac{R(S^\star)}{r(m-1)}(r(m-1)-\ell_1) - \ell_0 \cdot (1-\frac{R(S^\star)}{r(m-1)}),
    \end{align*}
    where the last term is the principal's utility under contract $t$.  
    Thus, $\utilmin{t}\leq \utilmin{\alpha}$, which concludes the proof of the theorem.
\end{proof}

\subsection{Optimal Linear Contracts in the General Model}
\label{sec:general}

Finally, we show that our computational results for the binary case translate to linear contracts in the general case.
For this, observe that when restricting attention to linear contracts, the agent's best response depends only on the expected rewards of action sets, not on their distributions. 
Therefore, finding the optimal contract among all contracts in the binary outcome model is equivalent to finding the optimal one  among all linear contracts in the general model.

\begin{corollary} The following hold in the general model:
    \begin{itemize}
        \item Given a gross substitutes function $R$ of expected rewards, and an additive cost function $c$, one can compute the optimal linear contract in polynomial time. 
        \item For a budget additive function $R$ of expected rewards, and an additive cost function $c$, it is NP-hard to compute the optimal linear contract.
        \item Given an arbitrary function $R$ of expected rewards, and an additive cost function $c$, one can compute a $(1-\epsilon)$-approximation to the  optimal linear contract in weakly polynomial time.
    \end{itemize}
\end{corollary}

\bibliographystyle{plain}
\bibliography{references}

\appendix

\section{Omitted Proofs}

\propzerozero*
\begin{proof}
    Consider a contract $t$ with $t(0) > 0$, and let $\set$ be the set of actions chosen by the agent under this contract.
    We distinguish between two cases.
    
    \textbf{Case 1:} $t(1)<t(0)$.
    In this case, consider the contract $t'$ where $t'(0)=t'(1)=0$. 
    We denote the set with all zero-cost actions by $\set_0$. By monotonicity of $f$ (and the tie-breaking of the agent in favor of the principal), under contract $t'$,  the agent selects a set of action $\set'$ with $f(\set')\geq f(\set_0)$. 
    It holds that 
    \begin{eqnarray*}
        t(0)-(t(0)-t(1))f(\set) & \geq & t(0)-(t(0)-t(1))f(\set)-c(\set) \\ & =& u_a(\set,t)\geq u_a(\set_0,t) \\ & =& t(0)-(t(0)-t(1))f(\set_0) -c(\set_0)\\ & = &t(0)-(t(0)-t(1))f(\set_0)  ,  
    \end{eqnarray*}
    which implies that $f(\set)\leq f(\set_0)\leq f(\set')$.
    It follows that $t'$ incentivizes the action set $\set'$, which has at least the same value of $f$ (as $\set)$ while paying (weakly) less, and therefore has at least the same principal's utility as $t$.

    \textbf{Case 2:} $t(1)\geq t(0)$.
    In this case, consider the contract $t'$ where $t'(0) = 0$ and $t'(1)=t(1)-t(0)$.
    Since for every action set $\set'$ it holds that 
    \begin{eqnarray*}
        u_a(\set',t')  & = & (t(1)-t(0)) f(\set')-c(\set') \\ & = & t(1) f(\set') +t(0)(1-f(\set'))-c(\set') -t(0) =u_a(\set',t)  -t(0),
    \end{eqnarray*}
    contract $t'$ incentivizes the same set $\set$ as $t$. Therefore, $t'$ incentivizes the same action set while paying (weakly) less, and therefore has at least the same principal's utility as $t$.
\end{proof}

\obscontright*
\begin{proof}
    The existence of the limits follows  immediately by the monotonicity of $\g_{f,c}$ (Corollary~\ref{cor:mon-g}).
    To prove the equality, let $\set$ be an arbitrary set in
    \[
        \argmin_{\set \in f^{-1}(\lim_{\epsilon\rightarrow 0^+}\g_{f,c}(\alpha + \epsilon))} {c(\set)}.
    \]
    We claim that $\set \in \demandset_{f,c}(\alpha)$.
    To see this, assume towards contradiction that there exists $\set'$ such that $\alpha \cdot f(\set') -c(\set') > \alpha \cdot f(\set) -c(\set)$.
    Let $\Delta_f = f(\set) -f(\set')$, and $\Delta_c = c(\set) -c(\set')$.
    First observe that $\Delta_f>0$. Indeed, $\Delta_f<0$ contradicts Proposition~\ref{prop:monotone}, and $\Delta_f=0$ contradicts the definition of $S$ as a set of minimal cost among sets with the same value.
    Similarly, since  $\Delta_f>0$, it also follows that  $\Delta_c>0$. 
	
    Let $\epsilon = \frac{\Delta_c}{2\Delta_f} -\frac{\alpha}{2}$. By the inequality above, $\epsilon >0$.
    We get 
    \begin{eqnarray*}
        (\alpha+\epsilon) \cdot f(\set') -c(\set') - [ (\alpha+\epsilon) \cdot f(\set) -c(\set) ] & = & -(\alpha+\epsilon) \cdot \Delta_f  + \Delta_c   \\ & =  & -\alpha\cdot \Delta_f   - \frac{\Delta_c}{2} + \frac{\alpha}{2} \cdot \Delta_f + \Delta_c  \\ & = &\epsilon \cdot \Delta_f > 0,
    \end{eqnarray*}
    which means that $\set'$ is strictly preferred by the agent over $\set$ at $\alpha+\epsilon$. 
    This contradicts Proposition~\ref{prop:monotone}.
\end{proof}

\begin{claim}\label{cl:submodular}
    Given a success probability function $f: \actionset \rightarrow [0,1]$, define a success probability function $g: 2^{\actions \cup \{n+1\}} \rightarrow [0,1]$ over actions $\actions \cup \{n+1\}$, where
    $$
        g(S)=
	   \begin{cases}
	       \beta_1 \cdot f(S) & \mbox{if } S \subseteq \actions\\
	       \beta_2 \cdot f(\actions) +f(S\setminus\{n+1\}) & \mbox{if } n+1 \in S
            \end{cases}
    $$
    for some $\beta_2 \geq \beta_1\geq 1$.
    If $f$ is a coverage function, then $g$ is a coverage function as well.
\end{claim}

\begin{proof}
    An equivalent definition for a function $f$ to be coverage is if there exist non-negative weights $\{w_T\}_{T\subseteq \actions}$ such that for every $S\subseteq \actions$, it holds that $f(S)=\sum_{T\subseteq A} w_T \cdot \mathbbm{1}_{S\cap T \neq \emptyset}$  (see Claim~\ref{cla:coverage} for details).   
    
    To show that $g$ is a coverage function, we construct weights $\{\hat{w}_T\}_{T \subseteq A\cup\{n+1\}}$ which satisfy for every set $S \subseteq A\cup \{n+1\}$ 
    \begin{equation}
        g(S)=\sum_{T\subseteq A \cup \{n+1\}} \hat{w}_T \cdot \mathbbm{1}_{S\cap T \neq \emptyset}.  \label{eq:coverage}
    \end{equation} 
    Let 
    $$
        \hat{w}_T=
	   \begin{cases}
	       w_T & \mbox{if } T \subseteq \actions\\
            (\beta_1-1) \cdot w_{T\setminus \{n-1\}}  & \mbox{if } n+1 \in T \mbox{ and } T \neq \{n+1\} \\ 
            (\beta_2-\beta_1+1) \cdot f(\actions)  & \mbox{if } T = \{n+1\}
	   \end{cases}
    $$
    To show that Equation~\eqref{eq:coverage} holds, we consider two cases.
    If $n+1 \notin S$ then 
    \begin{align*}
        g(S) &= \sum_{T\subseteq \actions \cup \{n+1\}} \hat{w}_T \cdot \mathbbm{1}_{S\cap T \neq \emptyset}\\ &=  \sum_{T\subseteq \actions } w_T \cdot \mathbbm{1}_{S\cap T \neq \emptyset} +\sum_{T\subseteq \actions } (\beta_1-1) \cdot w_T \cdot \mathbbm{1}_{S\cap T \neq \emptyset}  \\&= \beta_1 \cdot f(S).
    \end{align*}
    If $n+1 \in S$ then
    \begin{align*}
        g(S) &= \sum_{T\subseteq \actions \cup \{n+1\}} \hat{w}_T \cdot \mathbbm{1}_{S\cap T \neq \emptyset} \\ &=  \sum_{T\subseteq \actions } w_T \cdot \mathbbm{1}_{S\cap T \neq \emptyset} +\sum_{T\subseteq \actions } (\beta_1-1) \cdot w_T  + (\beta_2-\beta_1+1) \cdot f(\actions)  \\ &=  f(S) + \beta_2 \cdot f(\actions).
    \end{align*}
    This concludes the proof.	
\end{proof}

\begin{claim}\label{cla:coverage}
    A function $f$ is coverage if and only if there are non-negative weights $\{w_T\}_{T\subseteq \actions}$ such that for every $S\subseteq \actions$ it holds that $f(S)=\sum_{T\subseteq A} w_T \cdot \mathbbm{1}_{S\cap T \neq \emptyset}$.
\end{claim}
\begin{proof}
    Recall that a function $f$ is coverage if there exists a finite set $U$, where every element $j \in U$ is associated with a weight $\hat{w}_j \in \reals_{\geq 0}$, and a function $g:\actions \rightarrow 2^U$ such that for every set $S \in 2^\actions$, $f(S)=\sum_{j \in \bigcup_{i\in S}g(i)}\hat{w}_j$.
    
    We first show the only if direction. We start by adding, for each $T \subseteq A$, one element $j_T$ to $U$ with weight $\hat{w}_{j_T} = 0$ and setting $g(a) \leftarrow g(a) \cup \bigcup_{T: a \in T} \{j_T\}$. Since in this first step we only added elements of weight zero this does not change $f$. Now we define $w_T$ for $T \subseteq A$ as follows: We collect all elements $j \in U$ that are covered by \emph{all} elements of $T$ (i.e., $j \in g(a)$ for all $a \in T$), and sum up their weights $\hat{w}_j$. Note that this way we assign each of the original elements $j$ in $U$ to the set $T = \{a \mid j \in g(a)\}$, while the additional weight zero elements that we added and the extension of $g$ ensure that for each $T \subseteq A$ there is at least one element (namely $j_T$) that is covered by \emph{all} elements of $T$. Together these two properties ensure that for every $S \subseteq A$ it holds that 
    $$
        f(S) = \sum_{j \in \bigcup_{i \in S} g(i)} \hat{w}_j = \sum_{T \subseteq A} w_T \cdot \mathbbm{1}_{S\cap T \neq \emptyset}. 
    $$

    The if direction follows by the following construction: Let $U=\{T \mid T \subseteq \actions\}$,  let $\hat{w}_T = w_T$ for all $T \in U$, and for every $a\in \actions$ let $g(a) =\{T\mid a\in T\}$. Then for every $S \subseteq A$ it holds that 
    $$
        f(S) = \sum_{T\subseteq A} w_T \cdot \mathbbm{1}_{S\cap T \neq \emptyset}= \sum_{T \in \bigcup_{i\in S}g(i)}\hat{w}_T,
    $$
    as claimed.
\end{proof}

We next show that there are instances with gross-substitutes functions that admit a quadratic number of critical points.
\begin{proposition}\label{prop:gs-lower-bound}
    For every $n$, there exists a principal-agent instance $f,c$ over action set $A$ of size $n$ where $|\orbit{f,c}| \geq \frac{n (n+1)}{2}$.
\end{proposition}
\begin{proof}
    Consider the following OXS function\footnote{The family of OXS functions is a sub-family of gross substitutes functions. An OXS function $f:A\rightarrow \mathbb{R}_{\geq 0}$ is defined by a bipartite graph $G=(U=A,V,E)$, and a weight function $w:E \rightarrow \mathbb{R}_{\geq 0}$, where for $S \subseteq A$, $f(S)$ is the weight of the maximum matching in the graph induced by $S\cup V$.}  which is defined by the following full bipartite graph $G=(U=A=[n],V=[n],U\times V)$, and weights $w$ where for $i,j \in [n]$, $w(i,j) = 2^{i-n\cdot j}$ (where for convenience of notation, let $w(0,j)=0$). For a set $S\subseteq A$, it is easy to see that the maximum matching in $S\cup V$ always matches $S$ to $[|S|] \subseteq V$; moreover, it matches higher nodes in $S$ to lower nodes in $V$ since $2^{i_1-j_1\cdot n} +2^{i_2-j_2\cdot n} > 2^{i_1-j_2\cdot n}+2^{i_2-j_1\cdot n} $ when $i_1>i_2$ and $j_1<j_2$. Thus, for a  set $S = \{i_1,\ldots,i_k\}$ where $i_1 >i_2>\ldots > i_k$, it holds that $f(S) = \sum_{j=1}^k 2^{i_j-j\cdot n}$.  
    For every  $i\in [n]$, let $c(i)=\frac{ 3^{i}}{3^{n^2}}$ (and for convenience of notation, let $c(0)=0$).
   
    For $ i,j \in [n]$ with $i+j\leq n+1$,  let  
    \begin{align*}
        \alpha_{i,j} = 
        \frac{c(i)-c(i-1)}{w(i,j)-w(i-1,j)} =  
        \begin{cases}
            \frac{2 \cdot 3^{i-1} \cdot 2^{n\cdot j}}{3^{n^2} \cdot  2^{i-1}} &\text{ if } i >1 \\
            \frac{ 3 \cdot 2^{n\cdot j}}{2\cdot 3^{n^2}}  &\text{ if } i=1
        \end{cases},
    \end{align*}
    and let $S_{i,j} = \left([n]\setminus [n+1-j]\right)\cup\{i\}$ (where for convenience of notation, let $S_{0,1}=\emptyset$, and let $S_{0,j} = S_{n+2-j,j-1}$, and $\alpha_{0,j} = \alpha_{n+2-j,j-1}$ for $j>1$).
    Observe that for all $j\in [n]$,  $\alpha_{1,j},\ldots, \alpha_{n+1-j,j}$ is a strictly increasing sequence, and that for all $j$, it holds that $\alpha_{n+1-j,j} = \alpha_{0,j+1} < \alpha_{1,j+1}$.
    A simple analysis of the greedy algorithm shows that $\demandset_{f,c}(\alpha_{i,j}) = \{S_{i,j},S_{i-1,j}\}$, since at iteration $k=1,\ldots, j-1$ it strictly prefers action $n+1-k$, 
    and at the $j$-th iteration, it is indifferent between selecting actions $i$ and $i-1$ (where for $i=1$, in the $j$-th iteration it is indifferent between selecting action $i$ and selecting nothing). Thus, each of these values of $\alpha$ is a critical value which concludes the proof.    
\end{proof}

\end{document}